\documentclass[a4paper, 12pt, reqno]{amsart}

\usepackage[english]{babel}
\usepackage[T1]{fontenc}
\usepackage{NV_MathCommandsPack}
\usepackage{NV_Environments_English}

\numberwithin{equation}{section}

\def \fCount{\calN}
\def \anyBC{\bigstar}
\def \groundEn{\frE}
\def \densEn{\calE}
\def \signStat{\sharp}

\title[Thermodynamic Limit for Interacting Particles in Random Media]
{The Existence of the Thermodynamic Limit\\ 
for the System of Interacting Quantum Particles\\ 
in Random Media}

\author{Nikolaj~A.~Veniaminov}
\thanks{The author is partially supported by the grant 
ANR-08-BLAN-0261-01.}
\address{Laboratoire Analyse Géométrie et Applications, 
Université Paris 13 Nord,
99 avenue Jean-Baptiste Clément, 
93430 Villetaneuse, France}
\email{veniaminov@math.univ-paris13.fr}

\date{\today}

\begin{document}

\maketitle

\begin{abstract}
The thermodynamic limit of the internal energy and the entropy 
of the system of quantum 
interacting particles in random medium is shown to exist under
the crucial requirements of stability and temperedness of interactions.
The energy turns out to be proportional 
to the number of particles and/or volume of the system 
in the thermodynamic limit.
The obtained results require very general assumptions on the random 
one-particle model.
The methods are mainly based on subadditive type inequalities.
\end{abstract}


\section{Introduction}

Since the fundamental work \cite{Anderson_1} of P.~W.~Anderson,
the theory of random Schrödinger operators has been 
an extensively studied field of mathematical physics.
The greatest attention has been paid since now to the one-particle 
approximation and we do not try to list here even the major works 
on this topic.

There are relatively few papers where finitely many particles are considered.
There is a series of papers by Michael Aizenman and Simone Warzel that 
generalize the techniques of fractional moments method
(see, for example, \cite{AizenmanWarzel_LocBoundsMultipartilce})
and another series of articles by Victor Chulaevsky, Yuri Suhov and their 
collaborators that make use of the multiscale analysis
\cite{ChulaevskySuhov_MultiparticleAndersonLoc, 
ChulaevskySuhovMonvel_DynLocMultiParticle}.
The common point of these works is that they consider the number of particles 
being fixed and study the infinite volume limit for such system.

The present paper is an attempt to give an insight of what happens if 
both the number of particles and the volume go to infinity together so that
the number of particles per unit volume is kept constant.

The same question has already been addressed by various authors 
in case of absence of background
potential, i.e., when one-particle propagation is given by pure Laplacian.
In this paper, we will frequently follow the framework developed by 
David Ruelle in \cite{Ruelle_StatMech}, though
the presence of random potential presents certain mathematical difficulty,
which we will explain later.
We would also like to refer to an outstanding article
\cite{Lieb_Lebowitz_ExistenceThDynElectNucl}
of Elliot H.~Lieb and J.~L.~Lebowitz, 
where Coulomb interactions (always in absence of background potential)
are treated.

The idea that the number of particles grows with the volume 
looks natural in the context of condensed matter physics.
As a reference real-world example consider a piece of metal or semiconductor.
A bigger piece should contain proportionally more electrons.
As macroscopic objects are composed of many atoms 
(Avogadro constant $N_A \approx 6 \times 10^{23} \mathrm{mol}^{-1}$), 
and thus, ions and electrons, 
it turns out that the corresponding mathematical notion is the thermodynamic 
limit. 
Its existence for thermodynamic  quantities, 
such as internal energy, free energy, calorific capacity, and so
on, is the mathematical verification of the fact that these quantities 
are extensive. The latter is barely assumed in physics but actually 
needs rigorous verification.

Let us briefly discuss the mathematical objects we study.
All the notions will be introduced later in full regularity.
Let
\begin{equation}
H_\omega = -\laplace_d + V_\omega
\label{eq:HomegaIntro}
\end{equation}
be the random Schrödinger operator that describes a
single quantum particle in random environment $V_\omega$.
Kinetic part $\laplace_d$ is $d$-dimensional Laplacian.
One may also consider magnetic Schrödinger operator 
or whatever, provided that a number of basic facts,
such as Wegner estimate (see Proposition \ref{prop:WegnerEstimate}),
from the theory of one particle random operators hold true.
Actually, the whole ideology of this paper is that we take 
one-particle operators for known and deduce on this base properties
for multiparticle operators.

The restriction of $H_\omega$ to the domain $\Lambda$ is denoted by
$H_\omega(\Lambda)$.
For one particle Hamiltonian as in \eqref{eq:HomegaIntro}, 
we define, with a slight abuse of notation,
the $n$ particle operator (restricted in physical space to 
domain $\Lambda$) with pair interactions potential $U$ by
\begin{equation*}
H_\omega(\Lambda, n) 
= -\laplace_{n d} + \sum_{i = 1}^n V_\omega(x^i) 
+ \sum_{i \ne j} U(x^i - x^j) \text{,}
\end{equation*}
where $x^i \in \Lambda$, $i = 1, \hdots, n$, are particles' coordinates.

Using the notations introduced above, the general question we want to 
understand is the behavior of $H_\omega(\Lambda, n)$ in the thermodynamic limit:
\begin{equation}\label{eq:ThDynLimIntro}
H_\omega(\Lambda, n)\; \to\; ?, \quad 
|\Lambda| \to \infty,\; n \to \infty,\; n / |\Lambda| \to \const \text{.}
\end{equation}
In this paper, we answer a much more modest question than 
\eqref{eq:ThDynLimIntro}.
Namely, let $\groundEn_\omega(\Lambda, n)$ be the ground state energy of
$H_\omega(\Lambda, n)$.
In Theorem \ref{th:enConv}, we show, in particular, 
that the ground state energy 
per particle admits the thermodynamic limit:
\begin{equation}\label{eq:EnFondLimIntro}
\exists\; \lim \frac{\groundEn_\omega(\Lambda, n)}{n}, \quad 
|\Lambda| \to \infty,\; n \to \infty,\; n / |\Lambda| \to \const \text{.}
\end{equation}
Moreover, the same theorem gives a bit more general result
that allows to scale on the eigenenergy number in the spectrum. 
Roughly, the eigenenergy number (counting function) 
should be of order of exponent of the 
number of particles to ensure the convergence.

Theorem \ref{th:entrConv} gives the reciprocal result interchanging roles
of energy and the counting function 
in the spectrum (the theorem is stated in terms of
entropy which is the logarithm of counting function).

The main tool we use to obtain our results is a modified version of
subadditive ergodic theorem (see Proposition \ref{propESubAdd}).
For instance, one may show that the ground state energy 
$\groundEn_\omega(\Lambda, n)$ is additive with respect to the pair 
$(\Lambda, n)$ up to an error term that can be taken into account.
To make use of subadditivity we follow the construction of D.~Ruelle 
\cite{Ruelle_StatMech}.
Nevertheless, significant modifications are made in the proof because of
the fact that instead of full translation invariance of free Laplacian, 
we have only the covariance property of the family of random operators.
In general,
we are only able to prove the convergence in $L^2$ with respect to randomness 
(see Theorem \ref{th:enConv} case (\ref{it:CaseL2})).
A stronger convergence in $L^1$ and almost surely is established for 
compactly supported interactions (Theorem \ref{th:enConv} 
case (\ref{it:CaseL1as})).

In the last part of the present paper, we consider the system of noninteracting 
fermions in random medium.
We show that nontrivial effects arise due to Fermi-Dirac statistics even
in absence of interactions. In particular, we give an exact expression for 
the limit \eqref{eq:EnFondLimIntro} in terms of the 
one particle density of states measure (see Theorem \ref{th:densEnFerm})
and we find an interesting relation with the Fermi energy.

The rest of the paper is organized as follows.
The model of interacting quantum particles in random media and 
the notion of thermodynamic limit are introduced in
Section \ref{sect:ModelAndNotations}.
The results (mainly on the existence of thermodynamic limit) constitute 
Section \ref{sect:MainResults}, 
followed by the proofs in Section \ref{sect:Proofs}.
In addition, Section \ref{sect:Proofs} uncovers some extra
properties of the energy density (see Subsections 
\ref{subsect:CriticalDensityOfParticles} and 
\ref{subsect:PropertiesOfEnergyDensity}).
The proofs themselves may be instructive as well.
In Section \ref{sect:FreeParticles}, 
simple calculations concerning the thermodynamic 
limit for vanishing interactions are provided.

The author is grateful to his thesis advisor, Prof.~Frédéric Klopp, 
for proposing the problem, for his constant interest and support, 
as well as many valuable discussions.

\section{Model and Notations}
\label{sect:ModelAndNotations}

\subsection{Model of  Interacting Quantum Particles in Random Media}
\label{subsect:modelInterQuantPart}

We consider a system of $n$ interacting quantum particles in a random medium.
The discrete and continuum cases are treated simultaneously and 
an explicit indication is given if a result is valid only for one setting.
In the discrete case, the configuration space is given by 
$\calV = \bbZ^d$ 
and for the continuous case by
$\calV = \bbR^d$. 
In uniform manner, the one-particle Hilbert space is given by
\begin{equation*}
\frH = \frH^1 = L^2(\calV) \text{.}
\end{equation*}
The $n$-particle Hilbert space definition depends on the statistics (physical nature of quantum particles).
The following statistics are considered.
\begin{enumerate}
 \item
  \textbf{The Maxwell - Boltzmann statistics.}
The particles are physically distinguishable and 
no restrictions are imposed on a multiparticle wavefunctions.
This model is suitable, in particular, for the description 
of heavy atomic nuclei, 
i.~e., for particles that exhibit classical properties.
The corresponding Hilbert space is given by
\begin{equation*}
\frH^n = \bigotimes_{j = 1}^n \frH =  L^2(\calV^{n}) \text{.}
\end{equation*}
 \item
  \textbf{The Bose - Einstein statistics}: the particles are bosons. 
The wavefunction is necessarily symmetric with respect to the permutations of coordinates:
\begin{equation*}
\frH^n_+ = \Sym^{n} \frH = L_{+}^2(\calV^n) \text{,}
\end{equation*}
where $\Sym$ is the symmetrised tensor product.
 \item
  \textbf{The Fermi - Dirac statistics}: they describe fermions.
Wavefunctions are restricted to the antisymmetric subspace
\begin{equation*}
\frH^n_- = \bigwedge_{j = 1}^n \frH = L_{-}^2(\calV^n) \text{,}
\end{equation*}
where $\bigwedge$ is the external product.
\end{enumerate}
$\frH_+^n$ and $\frH_-^n$ are proper subspaces of $\frH^n$. 
For $\signStat \in \{\varnothing, +, -\}$, we
write $P_\signStat$ to denote the orthogonal projector on $\frH^n_\signStat$, 
where
$\signStat = \varnothing$ stands for the Maxwell - Boltzmann statistics,
$\signStat = +$ for the Bose - Einstein statistics and
$\signStat = -$ for the Fermi - Dirac statistics.
Obviously, 
$P = P_\varnothing = \ds1_{\frH^n}$ 
is the trivial projector.

One particle Hamiltonian is given by
\begin{equation*}
 H_\omega = H_\omega(1) = -\laplace + V_\omega \text{,}
\end{equation*}
and acts on $\Dom(H_\omega) \subset \frH$, where
\begin{itemize}
 \item 
$\laplace$ is either discrete or continuous Laplacian,
 \item
a random potential $V_\omega$ is (at least) $\bbZ^d$-ergodic 
and satisfies a decorrelation (independence at a distance) condition: 
$$
\exists R_0 > 0, \text{ such that if }
\dist(A, B) > R_0, \text{then }
\{V_\omega(x)\}_{x \in A} \text{ and } \{V_\omega(x)\}_{x \in B}
\text{ are independent.}
\eqno{\bf(IAD)}
$$
\end{itemize}

\begin{remark}
We also take into account the classes of random potentials that have the
ergodic group reacher than $\bbZ^d$-translations.
For instance, everything what follows remains true for the Poisson model.
\end{remark}

\begin{notation}\label{not:tauGamma}
We write $\Omega$, $\bbP$ and $\bbE$ for the associated probability space, 
probability measure and expectation respectively.
For $\gamma \in \bbZ^d$,
we denote by $\tau_\gamma$ the corresponding translations (measure preserving 
transformations) in $\Omega$
and by $T_\gamma$ the corresponding unitary transformations (coordinate 
shifts) in $L^2(\calV)$.
Namely, 
\begin{equation}\label{eq:covarRel}
H_{\tau_\gamma(\omega)} = T_\gamma^\ast H_\omega T_\gamma
\text{,}
\end{equation}
where $T_\gamma f(x) = f(x - \gamma)$, $f \in L^2(\calV)$, $x \in \calV$.
\end{notation}

By $H_\omega^{(i)}$ we denote 
a corresponding operator in $\frH^n$ that acts only on the $i$-th particle. 
More precisely,
\begin{equation}\label{eq:HomegaiDef}
 H_\omega^{(i)} = \underbrace{\ds1_\frH \otimes \hdots \otimes \ds1_\frH}_{\mbox{$i - 1$ times}} \otimes H_\omega \otimes \underbrace{\ds1_\frH \otimes \hdots \otimes \ds1_\frH}_{\mbox{$n - i$ times}} 
\text{.}
\end{equation}
The $n$-particle Hamiltonian in random environment $V_\omega$ 
and with interactions $W$ is given by the following self-adjoint operator 
on $\frH^n_\signStat$:
\begin{equation}\label{eqHomeganDef}
 H_{\omega, \signStat}(n) 
= P_\signStat \left[\sum_{i = 1}^n H_\omega^{(i)} + W_n\right] \text{.}
\end{equation}
For each $n \in \bbN$, $W_n$ is an interaction potential given 
by a function of the $n$ particles coordinates 
$\bfx = (x^1, \hdots, x^n)$, $x_j \in \calV$. 
We refer to the whole collection
$W = \{W_n\}_{n \in \bbN}$ as \emph{interactions} in general.
Remark also that in this model interactions are deterministic and all 
particles live in the same random background potential $V_\omega$.

In \eqref{eqHomeganDef}, the free part
\begin{equation*}
H^0_{\omega, \signStat}(n) = \sum_{i = 1}^{n} H_\omega^{(i)}
\end{equation*}
is called \emph{the second quantization} of 
$H_\omega$ in context of the Fock space 
(see, for example, \cite{BratteliRobinson2}).
Namely, we have to restrict the second quantization of $H_\omega$ 
to the $n$-particle subspace of the whole Fock space:
\begin{equation*}
H^0_{\omega, \signStat}(n) 
= \bigl.d\Gamma(H_\omega)\bigr|_{\frH^n_\signStat} \text{,}
\end{equation*}
where $d\Gamma$ denotes second quantization procedure.

\begin{remark}
$H^0_{\omega, \signStat}(n)$ 
acts from $\frH^n_\signStat$ into itself for any choice of $\signStat$,
whereas an arbitrary interaction potential $W_n$ does not necessarily preserve 
complete (anti)symmetry.
That is why the projector $P_\signStat$ a-priori acts non trivially in this 
formula.
However, potentials that we consider later are permutation symmetric 
(confer Section \ref{sect:MainResults}, property \textbf{(PI)}), 
so that the projector becomes obsolete in \eqref{eqHomeganDef}, i.e.,
\begin{equation*}
H_{\omega, \signStat}(n) = H^0_{\omega, \signStat}(n) + W_n \text{.}
\end{equation*}
\end{remark}

The Dirichlet and Neumann restrictions of $H_{\omega, \signStat}(n)$ 
to a finite box $\Lambda \subset \calV$ are denoted by
$H_{\omega, \signStat}^\anyBC(\Lambda, n)$,
where $\anyBC \in \{D, N\}$.
$H_{\omega, \signStat}^\anyBC(\Lambda, n)$ is a self-adjoint operator on
$\frH_\signStat^n(\Lambda) = L_\signStat^2(\Lambda^n)$.
We omit $\anyBC$ in notations frequently.

The operator $H_{\omega, \signStat}(\Lambda, n)$ has a discrete spectrum.
We call the \emph{counting function} associated to this operator
\begin{equation*}
\fCount_{\omega, \signStat}(E, \Lambda, n) =
\card\{E_k(\Lambda, n, \omega; \signStat) \leq E\} \text{,}
\end{equation*} 
where $E_k(\Lambda, n, \omega; \signStat)$ are the eigenvalues of $H_{\omega, \signStat}(\Lambda, n)$.
For the reasons that will become apparent later, 
\emph{the entropy} is a more convenient quantity:
\begin{equation}\label{eqSdef} 
S_{\omega, \signStat}(E, \Lambda, n) 
= \log \fCount_{\omega, \signStat}(E, \Lambda, n) \text{.}
\end{equation}

\begin{notation}
Sometimes we will drop some (if not all) of the indices and arguments
of the counting function and the entropy.
For example, if we are interested in the dependence on energy,
we will write just:
\begin{equation*}
S(E) = S_{\omega, \signStat}(E, \Lambda, n), \quad
\fCount(E) = \fCount_{\omega, \signStat}(E, \Lambda, n) \text{.}
\end{equation*}
\end{notation}

\begin{remark} 
As the counting function $\fCount$ takes its values in $\bbN \cup \{0\}$,
the entropy takes its values in 
$\log \bbN \cup \{-\infty\}$.
\end{remark}

\begin{observation}
For fixed $\omega$, $\signStat$, $\Lambda$ and $n$, 
the entropy $E \mapsto S_{\omega, \signStat}(E, \Lambda, n)$ 
is a non-decreasing right-continuous step function.
\end{observation}

The monotonicity of $S(E) = S_{\omega, \signStat}(E, \Lambda, n)$ 
allows to define a (quasi-)inverse function
$E_{\omega, \signStat}(\Lambda, n, S)$.
As $\fCount(\cdot, \Lambda, n)$ is not a local bijection at any point, 
the inverse function doesn't exist in a canonical manner. 
Our choice of the inverse is the following.

For $S$ such that $e^S \in \bbN$ we define
\begin{equation}\label{eqEdef} 
E_{\omega, \signStat}(\Lambda, n, S) =
E_{\exp{S}}(\Lambda, n, \omega; \signStat) \text{.}
\end{equation}
The application $S \mapsto E(S)$ is a right inverse 
of the entropy \eqref{eqSdef} in the following meaning.
For $S \in \log \bbN$ one has
\begin{equation}\label{eqSEinversion} S_{\omega, \signStat}\left(E_{\omega,
\signStat}(\Lambda, n, S), \Lambda, n\right) = S \text{.}
\end{equation}
Reciprocally, if $E \geq E_1(H_{\omega, \signStat}(\Lambda, n))$, then
\begin{equation}\label{eq:ESinversion}
E_{\omega, \signStat}\left(\Lambda, n, S_{\omega, \signStat}(E, \Lambda, n)\right) 
= E^- \text{,}
\end{equation}
where $E^-$ is the closest from below to $E$ 
eigenenergy of $H_{\omega, \signStat}(E, \Lambda, n)$.

The relations \eqref{eqSEinversion} and \eqref{eq:ESinversion}
motivate this choice of an inverse function.

\begin{definition} 
We denote by
$\groundEn = \groundEn_{\omega, \signStat}(\Lambda, n)$
the ground state energy of the operator
 $H_{\omega, \signStat}(\Lambda, n)$:
\begin{equation*} 
\groundEn_{\omega, \signStat}(\Lambda, n) =
\inf_{\substack{\varphi \in \Dom(H_{\omega, \signStat}(\Lambda, n))\\
\varphi \ne 0}} \frac{\left\langle H_{\omega, \signStat}(\Lambda, n)
\varphi, \varphi \right\rangle}{\|\varphi\|^2}
\text{.}
\end{equation*}
\end{definition}

Two characterizations of the ground state energy in terms of entropy 
are given below.

\begin{proposition}\label{propEFondCar}
$\groundEn$
is the ground state energy if and only if 
$\fCount(\groundEn - 0) = 0$ 
and
$\fCount(\groundEn + 0) > 0$
or, equivalently, if and only if 
$S(\groundEn - 0) = -\infty$ 
and 
$S(\groundEn + 0) \geq 0$.
\end{proposition}

\begin{proposition}
Alternatively, the ground state energy is given by the zero entropy: 
\begin{equation*}
 \groundEn_{\omega, \signStat}(\Lambda, n) = E_{\omega, \signStat}(\Lambda, n, 0) \text{.}
\end{equation*}
\end{proposition}
The latter characterization is essentially due to our choice of 
the inverse function $E(S)$ given by (\ref{eqEdef})
and would not be valid for another choice of the inverse,
whereas the Proposition \ref{propEFondCar} is universal with respect to
the particular choice of the function $E(S)$.

\subsection{Thermodynamic Limit}

In this section we discuss the notion of thermodynamic limit,
following the approach of \cite{Ruelle_StatMech}. 
For sake of completeness and the ease of reading, 
we repeat here the basic definitions related to the notion of
thermodynamic limit that can be found in various monographs and articles
such as
\cite{Ruelle_StatMech, Lieb_Seiringer_BoseGasCondensation, 
Griffiths_MicrocanonicalInStatMech, Lieb_Lebowitz_ExistenceThDynElectNucl}.

First of all, we give a precise meaning to the notion of 
a sequence of domains tending to infinity.

\begin{definition}
Let  $\diam(\Lambda)$ be the diameter of $\Lambda$ and 
$\partial_h\Lambda$ be the $h$-neighborhood of $\partial\Lambda$, i.e.,
\begin{equation*}
\partial_h\Lambda = \partial\Lambda + B(0, h) \text{,}
\end{equation*}
where $B(0, h)$ is the open ball of center $0$ and radius $h$.
\end{definition}

\begin{definition}\label{defFisherConv}
The sets $\Lambda$ tend to infinity \emph{in the sense of Fisher} if
\begin{equation*}
 \lim |\Lambda| = +\infty
\end{equation*}
and there exists a ``shape function'' $\pi$ such that
\begin{equation*}
\lim_{\alpha \to 0} \pi(\alpha) = 0
\end{equation*}
and for sufficiently small $\alpha$ and all $\Lambda$
\begin{equation*}
\left.|\partial_{\alpha \diam(\Lambda)}\Lambda| \right/
|\Lambda| \leq \pi(\alpha) \text{.}
\end{equation*}
\end{definition}
In what follows, 
we will always assume that $\Lambda \to \infty$ in the sense of Fisher.

\begin{remark}
Consider a sequence of rectangular domains. 
The fact that they tend to infinity in the sense of Fisher is equivalent
to say that all their sides tend to infinity at a comparable speed, i.e.,
\begin{equation*}
 \prod_{j = 1}^d [0, L_j] \to \infty \quad \Leftrightarrow \quad 
\left\{\begin{array}{l}
        \min_j L_j \to \infty,\\
        1 \geq \min_j L_j / \max_j L_j > 1 / C \text{.}
       \end{array}
\right.
\end{equation*}
\end{remark}

\begin{definition}\label{def:LimThd}
The limit 
$\Lambda \to \infty$, $n / |\Lambda| \to \rho$,
where $\rho$ is a positive constant (density of particles),
is called the \emph{thermodynamic limit}.
\end{definition}

Usually one is interested in extensive quantities per particle
or per unit of volume 
(that is the same thing up to a multiplicative constant due to
Definition \ref{def:LimThd})
while considering the thermodynamic limit.

\begin{definition}
Let $X_\omega(\Lambda, n; \calP)$ 
be a random variable that depends on 
a domain $\Lambda$, 
a number of particles $n$
and on a set of parameters $\calP$. 
We say that 
$X_\omega(\Lambda, n, \calP)$ \emph{admits the thermodynamic limit} 
if the limit
\begin{equation*}
 \lim_{\substack{\Lambda \to \infty\\n / |\Lambda| \to \rho\\\frL[\calP]}} \frac{X_\omega(\Lambda, n, \calP)}{n}
\end{equation*}
exists in some sense with respect to randomness $\omega$
(almost sure, in probability, in $L^2$).
Here $\frL[\calP]$ is a certain limiting procedure for the parameters $\calP$,
i.e., it determines the way how the parameters $\calP$ evolve 
when $\Lambda$ and $n$ go to infinity in the thermodynamic limit.
For example, see (\ref{eq:conv}), where an extra parameter is entropy $S$, and 
the limiting procedure for the entropy reads as it should tend to 
infinity linearly with the number of particles and/or the volume of the 
system.
\end{definition}

In thermodynamics, some commonly used quantities (such as internal energy, 
for example) are assumed to be extensive, i.e., additive with respect to
volume.
The existence of the thermodynamic limit
is the mathematically rigorous way of verifying the above assumption.
Thus, it is one of the fundamental questions of statistical physics.
Some authors go even further and refer to the question of existence of
thermodynamic limit purely as ``existence of thermodynamics'' 
\cite{Lieb_Lebowitz_ExistenceThDynElectNucl}.

In what follows, we will be primarily concerned with the existence of 
the thermodynamic limit for the energy 
$E_\omega(\Lambda, n, S)$ with
$S / n \to \sigma \geq 0$ and, in particular, 
the ground state energy $\groundEn_\omega(\Lambda, n)$,
i.e., for $\sigma = 0$.

\section{Main Results}
\label{sect:MainResults}

Throughout this section we work with Dirichlet boundary conditions
\begin{equation*}
H_\omega(\Lambda, n) = H_\omega^D(\Lambda, n)
\end{equation*}
and we omit the explicit indication $D$ in notations.
We give a series of statements concerning the existence of 
the thermodynamic limit for the model of interacting quantum particles 
in random media, which was introduced 
in Section \ref{subsect:modelInterQuantPart}.
Basic properties of the thus defined limits are discussed.

We shall need some assumptions on the model that we introduce now.

\textbf{Pair translation invariant interactions.}
The interactions are \emph{by pairs} and are 
\emph{invariant under translations}
if for all $n \in \bbN$
$$
 W_n(\bfx) = \sum_{1 \leq i < j \leq n} U(x^i - x^j) \text{,}
\eqno{\bf(PI)}
$$
where $U$ is a function on $\calV$.
We also assume that pair interactions are symmetric: 
$U(x) = U(-x)$, $x \in \calV$.

\textbf{Tempered interactions.}
Assume {\bf(PI)} and that there exist $R_0 > 0$, $A$ and $\lambda > d$
such that for all $|x| \geq R_0$
$$
|U(x)| \leq A |x|^{-\lambda} \text{.}
\eqno{\bf(PTI)}
$$
This condition 
(together with an additional assumption that $U$ is integrable in a 
neighborhood of zero) 
guarantees that interactions are of short range, i.e.,
\begin{equation*}
\int_{\bbR^d} |U(x)| {\rmd}x < +\infty \text{.}
\end{equation*}

The temperedness or similar conditions on the behavior of the interactions 
at the infinity have been used by various authors such as
Léon van~Hove, 
Joel L.~Lebowitz, 
Robert B.~Griffits and, in particular, 
Michael E.~Fisher and 
David Ruelle.
The reader is referred to \cite{Fisher_FreeEnergy}, \cite{Ruelle_StatMech},
\cite{FisherRuelle_StabManyParticles}, \cite{Lebowitz_StatMech},
\cite{Griffiths_MicrocanonicalInStatMech}.

\begin{remark}\label{rem:tempInter}
The above assumption of temperedness of interactions can be physically
motivated by the following argument.
Consider electrons in metal or semiconductor as a reference system.
Though electrons interact via Coulomb potential ($\sim 1 / r$) in vacuum,
the situation is different in metal where each electron is surrounded by
a ``cloud'' of other electrons and lives in a grid of ions.
This leads to what is called screening of Coulomb potential in metal 
(see \cite{AshcroftMermin_SSP, Zagoskin_QTofManyBodySyst})
and results to the effective interaction potential of the 
form
\begin{equation}\label{eq:YukawaPotential}
U(r) = \frac{Q}{r} \exp(-r / \lambda) \text{.}
\end{equation}
The interaction is between quasiparticles ``electron+cloud'', that are
called plasmons.\footnote{The potential 
\eqref{eq:YukawaPotential} is called Yukawa 
potential, though it usually arises in a context of nuclear physics.}
\end{remark}

\textbf{Lower-bounded one particle Hamiltonian.}
The one-particle random operator is bounded from below
uniformly with respect to randomness $\omega$:
$$
\exists C > 0, \text{ such that } H_\omega \geq -C 
\text{ for all } \omega \in \Omega \text{.}
\eqno{\bf(LB)}
$$




\begin{notation}
We write $\bbN_n = \{1, \hdots, n\}$.
For an index set $I = \{i_1, \hdots, i_n\} \subset \bbN$
we write
\begin{equation*}
x^I = (x^{i_1}, \hdots, x^{i_n}) \in \bbR^{n d}
\end{equation*}
for the vector of the coordinates of the particles enumerated by $I$,
where the elements are ordered in a nondecreasing fashion: 
$i_p < i_q$ if $p<q$.
\end{notation}

\begin{definition}\label{def:InterTerm}
Let $I_1 \cup I_2 = \bbN_{n_1 + n_2}$, $|I_j| = n_j$, 
be a partition of $n_1 + n_2$ particles in two disjoint subsets.
The \emph{term of interaction} 
between the particles $I_1$ and $I_2$ is given by
\begin{equation*}
W_{I_1, I_2}(x^{\bbN_{n_1 + n_2}}) = 
W_{n_1 + n_2}(x^{\bbN_{n_1 + n_2}}) - 
W_{n_1}(x^{I_1}) - W_{n_2}(x^{I_2}) \text{.}
\end{equation*}
\end{definition}

\textbf{Repulsive interactions.}
The interactions are \emph{repulsive}, if
for all $I_1$, $I_2$ as in Definition \ref{def:InterTerm}
it holds
$$
W_{I_1, I_2} \geq 0 \text{.}
\eqno{\bf(Rep)}
$$
If one assumes {\bf(Rep)} and that there are no self-interactions: 
$W_1 = 0$, then for all $n \in \bbN$
\begin{equation}\label{eq:RepInt1}
W_n(x^{\bbN_n}) \geq W_{n - 1}(x^{\bbN_{n - 1}}) + W_1(x^n) \geq 
W_{n - 1}(x^{\bbN_{n - 1}}) \geq \hdots \geq 0, \quad
x^{\bbN_n} \in \bbR^{n d} \text{.}
\end{equation}
If one also assumes {\bf(PI)}, then {\bf(Rep)} is equivalent to say that
$$
U \geq 0 \text{.}
$$

\textbf{Stable interactions.}
The interactions are stable if there exists $B > 0$,
such that for all $n \in \bbN$
$$
W_n(\bfx) \geq -n B \text{.}
\eqno{\bf{(SI)}}
$$
By \eqref{eq:RepInt1}, repulsive interactions are stable with $B = 0$.
The stability of interactions for various models is widely discussed, 
in particular, in \cite{FisherRuelle_StabManyParticles}.

\textbf{Compactly supported interactions.}
Using the notations of Definition \ref{def:InterTerm}, 
the interactions $W$ have \emph{compact support} if there exists $R_0 > 0$ 
such that 
$$
W_{I_1, I_2}(x^{\bbN_{n_1 + n_2}}) = 0
\eqno{\bf(Comp)}
$$
for all $x^{\bbN_{n_1 + n_2}} \in \calV^{n_1 + n_2}$ such that 
$\dist(x^{I_1}, x^{I_2}) \geq R_0$.

\begin{remark}
Obviously, for pair interactions, compact support is stronger 
than temperedness, i.e.,
\begin{equation*}
\mathbf{(PI)} + \mathbf{(Comp)} \Longrightarrow \mathbf{(PTI)} 
\text{ with } A = 0 \text{.}
\end{equation*}
\end{remark}

Let us now discuss the physical validity of the above assumptions.
For more details on classical electrodynamics, see, for example 
\cite{Jackson_ClassicalElectrodynamics} and 
for the electrodynamics of continuous media, 
see, for example \cite{LandauLifschitz8}.

\begin{itemize}
\item
The model of pair translation invariant {\bf(PI)} repulsive {\bf(Rep)} 
interactions
is natural for a description of identical quantum particles such as 
electrons.
\item
The condition of temperedness {\bf(PTI)} might seem more restrictive 
at first glance, but is usually circumvented as described in Remark 
\ref{rem:tempInter} by replacing actual interactions by screened 
interactions and bare electrons by quasiparticles.
\item
The condition of compactly supported interactions {\bf{(Comp)}} is a 
technical one and allows us to treat interaction of higher order than pair
(triple, etc.). However, even short range Yukawa interactions 
(\ref{eq:YukawaPotential}) are not compactly supported.
\item
The repulsive nature of interactions between identical particles 
{\bf{(Rep)}} is widely accepted.
Though, mathematically only the condition of stability {\bf{(SI)}} is 
needed.
Further discussion of stability condition and examples of catastrophic, 
i.e., not stable, potentials may be found in \cite{Ruelle_StatMech}.
\item
Finally, the lower boundedness of the one-particle operator {\bf{(LB)}} 
seems a natural basic assumption.
\end{itemize}

The following theorem is the main result of this paper on 
the existence of thermodynamics for the model described in 
Section \ref{sect:ModelAndNotations}.

\begin{theorem}[existence of thermodynamic limit]
\label{th:enConv}
Suppose that the one particle operator is lower bounded {\bf(LB)} and 
that the interactions are stable {\bf(SI)}.
Let also any of the following two cases hold:
\begin{enumerate}
\item\label{it:CaseL2}
interactions are translation invariant and by pairs, i.e., they satisfy
{\bf(PTI)}
\item\label{it:CaseL1as}
interactions are compactly supported, i.e., they satisfy {\bf(Comp)}.
\end{enumerate}
Then, the energy per particle admits thermodynamic limit, namely
\begin{equation}\label{eq:conv}
\frac{E_\omega(\Lambda, n, S)}{n}
\to \densEn(\rho, \sigma)
\qquad \text{as }
\Lambda \to \infty, 
\frac{n}{|\Lambda|} \to \rho,
\frac{S}{n} \to \sigma \text{,}
\end{equation}
where $\rho > 0$ and $\sigma \geq 0$. 
The convergence takes place in $L^2(\Omega)$ in case \eqref{it:CaseL2} and
in $L^1(\Omega)$ and $\omega$-almost sure in case \eqref{it:CaseL1as}.
The \emph{limiting energy density} $\densEn(\rho, \sigma)$ 
is defined by \eqref{eq:conv}, is a non-random function (does not depend 
on $\omega$) and
the limit is the same if both conditions 
\eqref{it:CaseL2} and \eqref{it:CaseL1as} are satisfied.
\end{theorem}

The energy density has the following basic properties.

\begin{proposition}[critical density of particles]
\label{prop:CritDens}
There exists \emph{a critical density} $\rho_c \in [0, +\infty]$ such that
\begin{equation*}
\begin{cases}
\densEn(\rho, \sigma) < +\infty,& \text{if $\rho < \rho_c$,}\\
\densEn(\rho, \sigma) = +\infty,& \text{if $\rho > \rho_c$,}
\end{cases}
\end{equation*}
for all $\sigma \geq 0$.
\end{proposition}

\begin{proposition}[energy density properties]
\label{prop:DensProp}
The energy density $\densEn(\rho, \sigma)$ is
\begin{enumerate}
\item
\label{it:densConv}
a convex function of variables $(\rho^{-1}, \sigma)$;
\item
\label{it:densMonot}
a nondecreasing function of $\rho$ and $\sigma$;
\item
\label{it:densCont}
a continuous function in the region 
$\{0 < \rho < \rho_c\} \times \{\sigma \geq 0\}$.
\end{enumerate}
\end{proposition}

\begin{corollary}
The energy density $\densEn(\rho, \sigma)$ admits an inverse 
$\sigma(\rho, \densEn)$. 
The latter is convex upwards with respect to
$(\rho^{-1}, \densEn)$ and is nondecreasing in $\densEn$ for any fixed $\rho$.
p\end{corollary}

Next we state a reciprocal result exchanging the roles of
energy and entropy
(the proof follows \cite{Griffiths_MicrocanonicalInStatMech}).

\begin{theorem}[existence of thermodynamic limit for entropy]
\label{th:entrConv}
Let the conditions of Theorem \ref{th:enConv} be satisfied.
Then for $0 < \rho < \rho_c$ and $\densEn \in \Ran{\densEn(\rho, \cdot)}$
\begin{equation*}\label{eq:convSigma}
\frac{S_\omega(E, \Lambda, n)}{n}
\to \sigma(\rho, \densEn)
\qquad \text{as }
\Lambda \to \infty, 
\frac{n}{|\Lambda|} \to \rho,
\frac{E}{n} \to \densEn \text{.}
\end{equation*}
The convergence takes place in the same sense 
as given by Theorem \ref{th:enConv}.
\end{theorem}

\begin{remark}
The condition that the energy belongs to the image of the function 
$\densEn(\rho, \cdot)$ is crucial.
One might remark as well that due to monotonicity and convexity properties 
of $\densEn$,
either $\densEn(\rho, \cdot) \equiv \const$ identically, or 
$\Ran{\densEn(\rho, \cdot)} = [\inf{\densEn(\rho, \cdot)}, +\infty)$.
\end{remark}

\section{Proofs}
\label{sect:Proofs}

This section is mainly devoted to the proof of $L^2$-convergence 
(case \eqref{it:CaseL2} of Theorem \ref{th:enConv}).
The basic ideas were inspired by \cite{Ruelle_StatMech} and 
\cite{Griffiths_MicrocanonicalInStatMech},
though the crucial difference is that instead of translation invariance of 
one particle operator (which is free Laplacian for both of the above works)
we have ergodicity, i.e., 
covariance with respect to a family of measure preserving transformations 
of the probability space.

We assume {\bf(LB)}, {\bf(PTI)} and {\bf(SI)} throughout this section,
except for Subsection \ref{subsect:L1conv},
where different assumptions will be made.
We also recall that the Dirichlet boundary conditions are used, i.e.,
$H_\omega = H_\omega^D$.

\subsection{Subadditive Inequalities}

Subadditive inequalities play the key role in our proofs.
The basic idea for all the proofs for existence theorems in this paper 
(and many others: see, for example, \cite{Ruelle_StatMech, 
Lieb_Lebowitz_ExistenceThDynElectNucl, Griffiths_MicrocanonicalInStatMech})
may be summarized as:
\begin{itemize}
\item
find a subadditive type inequality, 
\item
use the existing or prove an analog of subadditive ergodic theorem 
that guarantees the convergence.
\end{itemize}

Next is the core lemma that gives the subadditivity of energy. 

\begin{lemma}[Test function construction]
\label{lemmeFonctTest}
Suppose \textbf{(IAD)} and {\bf(PTI)} are satisfied.
Let the statistics $\signStat \in \{\varnothing, +, -\}$ be fixed. 
Consider domains $\Lambda_1$, $\Lambda_2$ 
such that $\dist(\Lambda_1,\Lambda_2) \geq r \geq R_0$
and functions 
$\varphi_j \in \frH_\signStat^{n_j}(\Lambda_j)$, $j = 1, 2$, 
with energies below $E_j$:
\begin{equation}\label{eq:phijBounds}
 \left\langle H_{\omega, \signStat}(\Lambda_j, n_j) \varphi_j, \varphi_j \right\rangle \leq E_j \|\varphi_j\|^2, \quad j = 1, 2 \text{.}
\end{equation}
Then, using $\varphi_1$, $\varphi_2$,
one can construct explicitly
$\zeta \in \frH_\signStat^{n_1 + n_2}(\Lambda_1 \cup \Lambda_2)$, 
a function of $n_1 + n_2$ particles
defined of a unified box $\Lambda_1 \cup \Lambda_2$ 
with energy below $E_1 + E_2 + A n_1 n_2 r^{-\lambda}$ : 
\begin{equation*}
 \left\langle H_{\omega, \signStat}(\Lambda_1 \cup \Lambda_2, n_1 + n_2) \zeta, \zeta \right\rangle \leq (E_1 + E_2 + A n_1 n_2 r^{-\lambda}) \|\zeta\|^2 \text{.}
\end{equation*}
\end{lemma}

\begin{remark}
The construction of a test function is explicit in the proof that follows.
\end{remark}

\begin{proof}[Proof of the Lemma]
We consider the extensions of the functions $\varphi_j$, $j = 1, 2$, 
by zero on $(\Lambda_1 \cup \Lambda_2)^{n_j}$, which we also denote by 
$\varphi_j$. 
Remark that \eqref{eq:phijBounds} implicitly contains the fact that 
$\varphi_j$ are zeros on the respective domains boundaries 
(due to Dirichlet condition)
so that the zero extension is a natural operation.
These extensions obviously preserve (anti)symmetry 
when $\signStat \in \{+, -\}$.
Consequently, 
one has $\varphi_j \in \frH_\signStat^{n_j}(\Lambda_1 \cup \Lambda_2)$ 
for any initial choice of $\signStat$.

We study each statistics separately now.

\begin{description}
\item[\textbf{Boltzmann statistics}]
Take
\begin{equation*}
 \zeta = \varphi_1 \otimes \varphi_2 \text{.}
\end{equation*}
Then, using {\bf(PTI)}
\begin{multline}\label{eqHnostatonZeta}
\bigl\langle H_{\omega}(\Lambda_1 \cup \Lambda_2, n_1 + n_2) \zeta, 
\zeta \bigr\rangle \\
\begin{aligned}
 &
\begin{split}
= \Bigl\langle\bigl\{ H_{\omega}(\Lambda_1, n_1) 
\otimes \ds1^{n_2} &+ \ds1^{n_1} \otimes H_{\omega}(\Lambda_2, n_2)\\
&+ W^{inter}(\Lambda_1, n_1; \Lambda_2, n_2)\bigr\} \varphi_1 
\otimes \varphi_2, \varphi_1 \otimes \varphi_2 \Bigr\rangle 
\end{split} \\
&
\begin{split}
\leq \left\langle H_\omega(\Lambda_1,n_1) \varphi_1, \varphi_1 \right\rangle
 \|\varphi_2\|^2 &+ 
\left\langle H_\omega(\Lambda_2,n_2) \varphi_2, \varphi_2 \right\rangle 
 \|\varphi_1\|^2 \\
&+ 
\left\langle 
\sum_{\substack{i = 1 \hdots n_1\\j = n_1 + 1 \hdots n_1 + n_2}}U(x^i - x^j) 
\varphi_1 \otimes \varphi_2, \varphi_1 \otimes \varphi_2 \right\rangle
\end{split} \\
&\leq (E_1 + E_2 + A n_1 n_2 r^{-\lambda}) \|\zeta\|^2 \text{,} &
\end{aligned}
\end{multline}
where $W^{inter}$ is the potential of interaction between 
$n_1$ particles in $\Lambda_1$ and $n_2$ particles in $\Lambda_2$
(see Definition \ref{def:InterTerm}).
In the last inequality, we used the temperedness of pair interaction 
potential $U$. 
As in each of $n_1 n_2$ terms, 
$x^i \in \Lambda_1$, $x^j \in \Lambda_2$
and $\dist(\Lambda_1, \Lambda_2) \geq r \geq R_0$, then by \textbf{(PTI)}
\begin{equation*}
\int_{\Lambda_1^{n_1}} \int_{\Lambda_2^{n_2}}\rmd{\bfx}
U(x^i - x^j) |\zeta(\bfx)|^2 \leq 
A r^{-\lambda} \int_{\Lambda_1^{n_1}} \int_{\Lambda_2^{n_2}}\rmd{\bfx}
|\zeta(\bfx)|^2 \leq 
A r^{-\lambda} \|\zeta\|^2 \text{,}
\end{equation*}
where we used the shorthand notation $\bfx = x^{\bbN_{n_1 + n_2}}$ 
for the vector of all particles' coordinates.

\item[\textbf{Bosons}]
We construct
\begin{equation}\label{eqZetaBosonicConstr}
 \zeta(x^{\bbN_{n_1 + n_2}}) 
= \sum_{\substack{I \subset \bbN_{n_1 + n_2}\\\card{I} = n_1}} 
\varphi_1(x^I) \varphi_2(x^{\bbN_{n_1 + n_2} \, \setminus \, I}) \text{.}
\end{equation}
This function is symmetric with respect to coordinate permutations.
Moreover, the terms in the sum \eqref{eqZetaBosonicConstr} 
are mutually orthogonal in 
$L^2\left((\Lambda_1 \cup \Lambda_2)^{n_1 + n_2}\right)$;
hence,
\begin{equation*}
 \|\zeta\|^2 = \binom{n_1 + n_2}{n_1} \|\varphi_1\|^2 \|\varphi_2\|^2 \text{,}
\end{equation*}
and the operator 
$H_{\omega, +}(\Lambda_1 \cup \Lambda_2, n_1 + n_2)$ 
preserves this orthogonality:
\begin{equation*}
H_{\omega, +}(\Lambda_1 \cup \Lambda_2, n_1 + n_2) \cdot 
\varphi_1(x^I) \varphi_2(x^{\bbN_{n_1 + n_2} \, \setminus \, I}) 
\; \bot \; \varphi_1(x^J) \varphi_2(x^{\bbN_{n_1 + n_2} \, \setminus \, J}) 
, \quad I \ne J \text{.}
\end{equation*}
Consequently, rewriting \eqref{eqHnostatonZeta} for $\zeta$ 
given by (\ref{eqZetaBosonicConstr})
and using the orthogonality,
one obtains: 
\begin{align*}
\left\langle H_{\omega, +}(\Lambda_1 + \Lambda_2, n_1 + n_2) 
\zeta, \zeta \right\rangle 
&\leq \binom{n_1 + n_2}{n_1} \cdot (E_1 + E_2 + A n_1 n_2 r^{-\lambda}) 
\cdot \|\varphi_1\|^2 \|\varphi_2\|^2 \\
&= (E_1 + E_2 + A n_1 n_2 r^{-\lambda}) \|\zeta\|^2 \text{.}
\end{align*}

\item[\textbf{Fermions}]
The construction for fermions is similar to that for bosons.
We define
\begin{equation}\label{eqZetaFermionicConstr}
 \zeta(x^{\bbN_{n_1 + n_2}}) 
= \sum_{\substack{I \subset \bbN_{n_1 + n_2}\\\card{I} = n_1}}
 (-1)^{\sum_{i \in I}i}  
\cdot \varphi_1(x^I) \varphi_2(x^{\bbN_{n_1 + n_2} \, \setminus \, I}) \text{.}
\end{equation}
This function is antisymmetric with respect to coordinate permutations.
The remaining part of the proof follows exactly that for the bosons.
\end{description}
\end{proof}

\begin{remark}
The construction \eqref{eqZetaFermionicConstr} is a generalization of 
the Slater determinant \cite{greiner2007quantum}.
\end{remark}

\begin{remark}
The Dirichlet boundary conditions are crucial for the proof
as they provide a zero cost (canonical) extension of functions from 
$\Dom{H_\omega(\Lambda_j, n_j)}$
to a larger domain $\Lambda_1 \cup \Lambda_2$
without changing the norm.
At this moment, 
we are not able to prove an analog of Theorem \ref{th:enConv} 
(essentially, we need an analog of Lemma \ref{lemmeFonctTest})
for Neumann or periodic boundary conditions.
\end{remark}

From now on, we omit the statistics sign $\signStat$ in the notations. 
Everything that follows is valid for all the statistics.
However, one should be warned that quantities (such as limiting values) 
may depend on the statistics.

\begin{proposition}
Let the interactions $W$ be tempered {\bf(PTI)}. 
If 
$\dist(\Lambda_1,\Lambda_2) \geq r \geq R_0$, 
then
\begin{align}
\label{eqNLowerEst}
 N_{\omega}(\Lambda_1 \cup \Lambda_2, n_1 + n_2, 
E_1 + E_2 + A n_1 n_2 r^{-\lambda}) 
&\geq N_{\omega}(\Lambda_1, n_1, E_1) N_{\omega, \eps}(\Lambda_2, n_2, E_2) 
\text{,}\\
\label{eqSLowerEst}
 S_{\omega}(\Lambda_1 \cup \Lambda_2, n_1 + n_2, 
E_1 + E_2 + A n_1 n_2 r^{-\lambda}) 
&\geq S_{\omega}(\Lambda_1, n_1, E_1) 
+ S_{\omega}(\Lambda_2, n_2, E_2) \text{.}
\end{align}
\end{proposition}

\begin{proof}
The proof of \eqref{eqNLowerEst} is done using 
the variational principle
for eigenvalues of 
$H_{\omega}(\Lambda, n)$ 
and the function $\zeta$ from Lemma \ref{lemmeFonctTest} 
as a test function. 
Taking the logarithm, one obtains \eqref{eqSLowerEst}.
\end{proof}

\begin{proposition}\label{propESubAdd}
Let the interactions $W$ be tempered {\bf(PTI)}.
\begin{enumerate}
 \item 
Take $S_1$, $S_2$ such that $\exp{S_i} \in \bbN$, $i = 1, 2$. 
If $\dist(\Lambda_1,\Lambda_2) \geq r \geq R_0$, then
\begin{equation}\label{eqEupperEst1}
 E_\omega(\Lambda_1 \cup \Lambda_2, n_1 + n_2, S_1 + S_2) 
\leq E_\omega(\Lambda_1, n_1, S_1) + E_\omega(\Lambda_2, n_2, S_2) 
+ A n_1 n_2 r^{-\lambda} \text{.}
\end{equation}
 \item
Take $S_i$ such that 
$\exp{S_i} \in \bbN$, $i = 1, \hdots, m$,
and domains $\Lambda_1, \hdots, \Lambda_m$ 
at mutual distances greater than $r \geq R_0$. Then
\begin{equation}\label{eqEupperEst}
 E_\omega\left(\bigcup_{i = 1}^m \Lambda_i, \sum_{i = 1}^m n_i, \sum_{i = 1}^m S_i\right) \leq \sum_{i = 1}^m E_\omega(\Lambda_i, n_i, S_i) + \frac{A}{2} \left(\sum_{i = 1}^m n_i\right)^2 r^{-\lambda} \text{.}
\end{equation}
\end{enumerate}
\end{proposition}

\begin{proof}
The inequality \eqref{eqEupperEst} is an immediate consequence of 
\eqref{eqEupperEst1}.
The latter is obtained by taking 
$E_\omega(\Lambda_1 \cup \Lambda_2, n_1 + n_2, \cdot)$ of
\eqref{eqSLowerEst} and using \eqref{eqSEinversion}.
\end{proof}

\subsection{$L^2$-convergence on a Special Sequence of Cubes}
\label{subsectSuiteCubes}

In this section, we will construct a special sequence of cubes 
$\Lambda_N$ in configuration space $\calV$, 
on which the existence of thermodynamic limit will be proven.
The idea is inspired by \cite{Ruelle_StatMech}.

Let $\theta$ be a number that satisfies
\begin{equation*}
 1 < 2^{d / \lambda} < \theta < 2
\end{equation*}
and let
\begin{equation}\label{eq:LDef}
 \wtL > R = \frac{R_0 + \delta}{2 - \theta} \text{,}
\end{equation}
where $\delta > 0$ is a constant that will be fixed later.
For an integer $N \geq 0$ put
\begin{equation}\label{eqLNcubeCote}
 L_N = 2 \left[\frac{1}{2} \left(2^N \wtL - \theta^N R\right)\right] \text{,}
\end{equation}
so that $L_N \in 2 \bbZ$,
and define the cube $\Lambda_N$ by
\begin{equation*}
 \Lambda_N = [-L_N / 2, L_N / 2]^d \subset \calV \text{.}
\end{equation*}
Remark that the vertices of $\Lambda_N$ are at integer points.
According to \eqref{eqLNcubeCote}
it is possible to place $2^d$ translates of 
$\Lambda_N$ (cubes $\Lambda_N^{(i)}$) inside $\Lambda_{N + 1}$ 
at mutual distances at least
\begin{equation*}
 R_N = L_{N + 1} - 2 L_N = \theta^N (2 - \theta) R + \eps 
= \theta^N (R_0 + \delta) + \eps \geq R_0 \text{,}
\end{equation*}
where $\eps \in [-4, 2]$ is the error due to the rounding procedure.
The constant $\delta$ is chosen to compensate a possibly negative error 
term $\eps$, so that the last inequality holds true. 
It suffices, for example, to choose $\delta = 4$.

We remark that cubes $\Lambda_N^{(i)}$ 
are explicitly given by
\begin{equation*}
 \Lambda_N^{(i)} = \Lambda_N + \gamma_i^{(N)} \text{,}
\end{equation*}
where
\begin{equation*}
 \gamma_i^{(N)} = \frac{L_{N + 1} - L_N}{2} \cdot e_i \in \bbZ^d,
\quad e_i = (\pm1, \hdots, \pm1) \in \bbR^d,\quad i = 1, \hdots, 2^d \text{.}
\end{equation*}

\begin{remark}
It is important that the translation vectors $\gamma_i^{(N)}$ are integer,
because it ensures that the restrictions of the random potential $V_\omega$ to 
$\Lambda_i^{(N)}$ for different $i$ are connected by the covariance 
relation \eqref{eq:covarRel}.
\end{remark}

The function 
$E_{\omega}(\Lambda, n, S)$ 
satisfies the following monotonicity properties.

\begin{lemma}\label{lem:Emonotone} 
For fixed $\omega$ and $n$, the energy $E_\omega(\Lambda, n, S)$ is
\begin{enumerate}
 \item a nondecreasing function of $S$,
 \item a nonincreasing function of $\Lambda$.
\end{enumerate}
\end{lemma}

By Lemma \ref{lem:Emonotone}
and the almost-subadditivity condition \eqref{eqEupperEst}
we obtain for $S_i$ in $\log{\bbN}$ that
\begin{align*}
 E_\omega\left(\Lambda_{N + 1}, \sum_{i = 1}^{2^d} n_i, 
\sum_{i = 1}^{2^d} S_i\right) 
&\leq 
E_\omega\left(\bigcup_{i = 1}^{2^d} \Lambda_N^{(i)}, 
\sum_{i = 1}^{2^d} n_i, \sum_{i = 1}^{2^d} S_i\right) \\
&\leq \sum_{i = 1}^{2^d} E_\omega(\Lambda_N^{(i)}, n_i, S_i) + 
\frac{A}{2} \left(\sum_{i = 1}^{2^d} n_i\right)^2 R_N^{-\lambda} \\
&= 
\sum_{i = 1}^{2^d} E_{\tau_{\gamma_i^{(N)}}(\omega)}(\Lambda_N, n_i, S_i) + 
\frac{A}{2} \left(\sum_{i = 1}^{2^d} n_i\right)^2 R_N^{-\lambda} \text{,}
\end{align*}
where $\left\{\tau_\gamma\right\}_{\gamma \in \bbZ^d}$ 
is the family of ergodic transformations of $\Omega$, that
were introduced in Notation \ref{not:tauGamma}.
In particular, for $S \in \log{\bbN}$
\begin{equation}\label{eqSousAdditSuite}
 E_\omega\left(\Lambda_{N + 1}, 2^d n, 2^d S\right) \leq 
\sum_{i = 1}^{2^d} E_{\tau_{\gamma_i^{(N)}}(\omega)}(\Lambda_N, n, S) + 
\frac{A}{2} (2^d n)^2 R_N^{-\lambda} \text{.}
\end{equation}
Let now $\rho$ and $\sigma$ be positive numbers
such that $2^{N_0 d} \rho \wtL^d$ and 
$\exp\left(2^{N_0 d} \sigma \rho \wtL^d\right)$ 
are integer, for a sufficiently large integer $N_0$. 
Plug in \eqref{eqSousAdditSuite}
\begin{equation}\label{eqnNSNdef}
 n_N = 2^{N d} \rho \wtL^d,\quad S_N = \sigma n_N = 2^{N d} \rho \sigma \wtL^d
\end{equation}
for $N \geq N_0$.
Remark that
\begin{equation*}
 n_N / |\Lambda_N| \to \rho,\quad 
 S_N / n_N = \sigma,\quad 
 2^{N d} \wtL^d / |\Lambda_N| \to 1\quad
\text{as } N \to +\infty \text{.}
\end{equation*}
We introduce the following sequence of random variables
\begin{equation}\label{eq:XNdef}
 X_N(\omega) = 2^{-N d} \left(E_\omega(\Lambda_N, n_N, S_N) + (B + C) n_N\right) 
\text{,}
\end{equation}
where $B$ is the constant from \textbf{(SI)} and
$C$ is the constant from \textbf{(LB)}.
By \eqref{eqSousAdditSuite} this sequence satisfies the inequality
\begin{equation*}
 X_{N + 1}(\omega) \leq 2^{-d} \sum_{i = 1}^{2^d} X_N(\tau_{\gamma_i^{(N)}}(\omega)) + G_N
\end{equation*}
with
\begin{equation*}
 G_N = A \rho^2 \wtL^{2 d} 2^{(N + 2) d -1} R_N^{-\lambda} \text{.}
\end{equation*}
In order to show the convergence of the sequence $X_N(\omega)$, 
we establish the following proposition.

\begin{proposition}\label{propDecrMoyenne}
Let $X_N$ be a sequence of nonnegative random variables 
on a probability space $(\Omega, \bbP)$, $X_N(\omega) \geq 0$, 
such that for each $N$ 
there exists a family of probability preserving transformations of $\Omega$, 
$\tau_i^{(N)}$, $i \in I_N$, with $\card I_N < +\infty$, 
such that the variables 
$X_N \circ \tau_i^{(N)}$, $i \in I_N$ are i.i.d. 
(independent identically distributed). 
If the sequence $X_N$ satisfies
\begin{equation}\label{eqDecroissMoyenne}
 X_{N + 1}(\omega) \leq 
\frac{1}{\card I_N} \sum_{i \in I_N} X_N(\tau_i^{(N)} \omega) + G_N,
\quad \text{where} \quad \sum_{N = 1}^{+\infty} |G_N| < +\infty \text{,}
\end{equation}
then there exists a constant $\overline{X}$ 
(that does not depend on $\omega$) such that
\begin{equation}\label{eqL2ConvSuite}
 X_N \xrightarrow[N \to +\infty]{L^2} \overline{X} \text{.}
\end{equation}
\end{proposition}

\begin{proof}
Since the terms at the r.h.s. of \eqref{eqDecroissMoyenne} 
are identically distributed, after taking the expectation 
one obtains:
\begin{equation}\label{eqPresqueDecroiss}
 \bbE X_{N + 1} \leq \bbE X_{N} + G_N \text{.}
\end{equation}
Consider the sequence
\begin{equation*}
 c_N = \bbE X_N - \sum_{i = 1}^{N - 1} G_i \text{.}
\end{equation*}
Obviously, \eqref{eqPresqueDecroiss} guarantees that $c_{N + 1} \leq c_N$. 
Consequently, this sequence converges: 
$c_N \xrightarrow[N \to \infty]{} c_0$. 
Thus, as soon as the sum $\sum_{i = 1}^{+\infty} G_i$ also converges, 
$\bbE X_N$ admits the limit
that we denote by $\overline{X}$:
\begin{equation*}
 \bbE X_N \xrightarrow[N \to +\infty]{} \overline{X} \text{.}
\end{equation*}
Consider now the variance of $X_N$, 
$\bbE(X_N - \bbE X_N)^2 = \bbE(X_N^2) - (\bbE X_N)^2$.
We will show that this variance tends to zero as $N$ goes to infinity. 
By \eqref{eqDecroissMoyenne}
\begin{align*}
 (X_{N + 1}(\omega))^2 
&\leq \left(\frac{1}{\card I_N} \sum_{i \in I_N} X_N(\tau_i^{(N)} \omega) 
+ G_N\right)^2 \\
&
\begin{aligned}
=
\frac{1}{(\card I_N)^2} \sum_{i \in I_N} \left(X_N(\tau_i^{(N)} \omega)\right)^2
&+ \frac{1}{(\card I_N)^2} \sum_{\substack{i, j \in I_N\\i \ne j}} 
X_N(\tau_i^{(N)} \omega) X_N(\tau_j^{(N)} \omega) \\
&+
\frac{2 G_N}{\card I_N} \sum_{i \in I_N} X_N(\tau_i^{(N)} \omega) + G_N^2 \text{.}
\end{aligned}
\end{align*}
Taking the expectation and using the fact that 
$\left(X_N(\tau_i^{(N)})\right)_{i \in I_N}$ are i.i.d. for fixed $N$, we find
\begin{align*}
&
\begin{aligned}
\bbE(X_{N + 1}(\omega))^2
\leq
\frac{1}{(\card I_N)^2} 
\sum_{i \in I_N} \bbE\left(X_N(\tau_i^{(N)} \omega)\right)^2 
&+ 
\frac{1}{(\card I_N)^2} \sum_{\substack{i, j \in I_N\\i \ne j}} 
\bbE\left(X_N(\tau_i^{(N)} \omega) X_N(\tau_j^{(N)} \omega)\right) \\
\phantom{\leq 
\frac{1}{(\card I_N)^2} 
\sum_{i \in I_N} \bbE\left(X_N(\tau_i^{(N)} \omega)\right)^2}
&+ \frac{2 G_N}{\card I_N} \sum_{i \in I_N} 
\bbE X_N(\tau_i^{(N)} \omega) + G_N^2 
\end{aligned}
\\
&
\begin{aligned}
\phantom{\bbE(X_{N + 1}(\omega))^2}
&= 
\frac{1}{(\card I_N)^2} \sum_{i \in I_N} \bbE(X_N^2) + 
\frac{1}{(\card I_N)^2} \sum_{\substack{i, j \in I_N\\i \ne j}} (\bbE X_N)^2
+ \frac{2 G_N}{\card I_N} \sum_{i \in I_N} \bbE X_N + G_N^2 \\
&= 
\frac{1}{\card I_N} \bbE(X_N^2) + 
\frac{(\card I_N)(\card I_N - 1)}{(\card I_N)^2} (\bbE X_N)^2 +
2 G_N \cdot \bbE X_N + G_N^2 \text{.}
\end{aligned}
\end{align*}
So, we get
\begin{equation}\label{eqVarianceEst}
 \bbE(X_{N + 1}^2) \leq 
\frac{1}{\card I_N} \bbE(X_N^2) + 
\left(1 - \frac{1}{\card I_N}\right) (\bbE X_N)^2 +
2 G_N \cdot \bbE X_N + G_N^2 \text{.}
\end{equation}
By the Schwarz inequality
\begin{equation}\label{eqInegSchwarz}
 (\bbE X_N)^2 \leq \bbE(X_N^2) \text{.}
\end{equation}
Hence, using \eqref{eqVarianceEst}, we obtain
\begin{equation*}
 \bbE(X_{N + 1}^2) \leq \bbE(X_{N}^2) + 2 G_N \sqrt{\bbE(X_{N}^2)} + G_N^2 = 
\left(\sqrt{\bbE(X_{N}^2)} + G_N\right)^2 \text{.}
\end{equation*}
Finally,
\begin{equation*}
 \sqrt{\bbE(X_{N + 1}^2)} \leq \sqrt{\bbE(X_{N}^2)} + G_N \text{.}
\end{equation*}
Arguing as for the expectation, this implies that the sequence
 $\bbE(X_{N}^2)$ converges. 
Taking the limit in \eqref{eqVarianceEst} 
and using the fact that 
$G_N \to 0$ as $N \to +\infty$, 
we obtain
\begin{equation*}
 \lim_{N \to +\infty} \bbE(X_{N}^2) \leq \overline{X}^2 \text{,}
\end{equation*}
but the Schwarz inequality \eqref{eqInegSchwarz} 
immediately gives the reciprocal estimate. We conclude that 
\begin{equation*}
 \lim_{N \to +\infty} \bbE(X_{N}^2) = \overline{X}^2
\end{equation*}
and the variance $\bbE(X_N - \bbE X_N)^2$ 
tends to $0$ as $N$ goes to infinity.
This proves \eqref{eqL2ConvSuite}.
\end{proof}

As an immediate consequence of this general statement 
we obtain the existence of the thermodynamic limit 
on the special sequence of cubes $\Lambda_N$.

\begin{corollary}\label{corL2LimSuite}
Suppose {\bf(LB)}, {\bf(PTI)} and {\bf(SI)}.
Then the thermodynamic limit 
for the energy $E_\omega$ on the sequence of cubes $\Lambda_N$ 
in the sense of $L^2_\omega$ exists, i.e.,
\begin{equation*}
 \frac{E_\omega(\Lambda_N, n_N, S_N)}{n_N} 
\xrightarrow[N \to +\infty]{L^2} \densEn(\rho, \sigma)  \text{,}
\end{equation*}
where $\densEn(\rho, \sigma)$ is defined by this limit 
and is called \emph{the limiting energy per particle} or
\emph{the energy density}.
\end{corollary}
\begin{proof}
We only need to prove that the random variable $X_N(\omega)$ introduced in
\eqref{eq:XNdef} is nonnegative.
Recall that
\begin{equation*}
H_\omega(\Lambda_N, n_N) = \sum_{i = 1}^{n_N} H_\omega^{(i)}(\Lambda_N, n_N) 
+ W_{n_N} \text{.}
\end{equation*}
Moreover (see also \eqref{eq:HomegaiDef}),
\begin{equation*}
H_\omega^{(i)}(\Lambda_N, n_N) \geq -C, \quad
i = 1, \hdots, n_N \text{,}
\end{equation*}
by the lower boundedness of the one particle Hamiltonian \textbf{(LB)}, and
\begin{equation*}
W_{n_N} \geq -B n_N
\end{equation*}
by the stability of interactions \textbf{(SI)}. Thus,
\begin{equation*}
H_\omega(\Lambda_N, n_N) \geq -(B + C) n_N \text{,}
\end{equation*}
and consequently
\begin{equation*}
E_\omega(\Lambda_N, n_N, S_N) + (B + C) n_N \geq 0 \text{.}
\end{equation*}
\end{proof}

\subsection{Critical Density of Particles}
\label{subsect:CriticalDensityOfParticles}

We now discuss the finiteness of the thermodynamic limit 
that was announced in Proposition \ref{prop:CritDens}.

Being essentially attained via a nonincreasing sequence, 
the limit is finite if and only if there are finite terms in the sequence 
$\dfrac{E_\omega(\Lambda_N, n_N, S_N)}{n_N}$.
In other words, if for a sufficiently big $N$,
$E_\omega(\Lambda_N, n_N, S_N)$ is finite, 
then $\densEn(\rho, \sigma) < +\infty$.

The situation when the operator $H_\omega(\Lambda_N, n_N)$ doesn't possess
an increasing sequence of eigenvalues may arise, 
according to variational principle, if and only if
a subspace of functions $\varphi$ such that
$(H_\omega(\Lambda_N, n_N) \varphi, \varphi) < +\infty$
is of finite dimension.
But this last condition is possible only if 
the interaction potential $W_n$ takes the value $+\infty$ and if
there are too few configurations with a finite interaction term, i.e.,
\begin{equation}\label{eqCondInfinInter}
 \meas\left\{(x^1, \hdots, x^{n_N}) \in \Lambda_N^{n_N} 
\mid W_n(x^1, \hdots, x^{n_N}) < +\infty\right\} = 0 \text{.}
\end{equation}

As a model case, suppose that the interactions are by pairs {\bf(PI)} and 
that the pair potential $U$ represents hard cores of radius $r_0$
(see \cite{Ruelle_StatMech}):
\begin{equation*}
 U(x) \begin{cases}
       = +\infty,& |x| \leq r_0,\\
       \ne +\infty,& |x| > r_0 \text{.}
      \end{cases}
\end{equation*}
In this case, the condition \eqref{eqCondInfinInter} is satisfied 
if there isn't enough space for $n_N$ balls of radius $r_0 / 2$ 
with centers in the domain $\Lambda_N$.
In other words, define the set of denied spacial configurations of
$n$ particles by
\begin{equation*}
S_{r_0}^n = \left\{(x^1, \hdots, x^n) \in \bbR^{n d}, \text{ such that } 
|x^i - x^j| < r_0 \text{ for some } i \ne j\right\} \text{.}
\end{equation*}
Then the Hamiltonian is defined on $L^2_\signStat(\calV^n \setminus S_{r_0}^n)$,
instead of $L^2_\signStat(\calV^n)$ and
it may happen that 
\begin{equation*}
\meas\left(\calV^n \setminus S_{r_0}^n\right) = 0 \text{.}
\end{equation*}
The last observation suggests
that there exists \emph{a critical density of particles} $\rho_c$ 
such that the energy density $\densEn(\rho, \sigma)$ is finite for 
$\rho < \rho_c$ and infinite for $\rho > \rho_c$.
For example, for the case of hard cores, this is the closed packing density.
Note that $\rho_c = +\infty$ if
the interaction potential takes only finite values.

\subsection{Properties of the Energy Density}
\label{subsect:PropertiesOfEnergyDensity}

Before proceeding with the proof of the existence of thermodynamic limit 
for general domains, we establish some properties of the 
energy density $\densEn$ 
(in particular, we prove Proposition \ref{prop:DensProp}).

From now on, we assume that $\rho < \rho_c$.
Note that till now the function $\densEn(\rho, \sigma)$ 
is defined (as a limit of a sequence) only for particle densities $\rho$ 
and entropy densities $\sigma$ of the form
\begin{equation}\label{eqRhoCond}
 \rho = \frac{m_1}{2^{N_0 d} L^d},\qquad \sigma = \frac{\log{m_2}}{m_1} \text{,}
\end{equation}
where $m_1$, $m_2$ and $N_0$ are positive integers.
In all the statements that follow we implicitly assume that 
$\rho$ and $\sigma$ satisfy \eqref{eqRhoCond}.

\begin{proof}[Proof of Proposition \ref{prop:DensProp} \eqref{it:densConv}]
The convexity of the limiting function is an immediate consequence 
of the almost-subadditivity \eqref{eqEupperEst1}.
\end{proof}

\begin{proof}[Proof of Proposition \ref{prop:DensProp} \eqref{it:densMonot}]
The monotonicity is given by Lemma \ref{lem:Emonotone}.
\end{proof}

\begin{proposition}\label{prop:DensLocalBound}
The energy density $\densEn$ is locally bounded
on the plane of parameters $0 < \rho < \rho_c$ and $\sigma \geq 0$.
\end{proposition}

\begin{proof}
Let $0 < \rho_1 < \rho_2 < \rho_c$ and $\sigma_0 > 0$ 
be of the form \eqref{eqRhoCond}. 
We shall show that $\densEn$ is bounded in the region 
$\Delta = 
\{\rho_1 \leq \rho \leq \rho_2\} \times \{0 \leq \sigma \leq \sigma_0\}$. 
First remark that the number of particles $n_N = 2^{N d} \wtL^d \rho$ 
with $\rho \leq \rho_2$ can be represented as
\begin{equation}\label{eqnNrepr}
 n_N = \sum_{j = 1}^{2^{N d}} n_0^{(j)} \text{,}
\end{equation}
where $n_0^{(j)} \in \left\{\left[\wtL^d \rho\right], 
\left[\wtL^d \rho\right] + 1\right\}$. 
Obviously, such a representation depends on $\rho$. 
On the other hand, the bound 
$n_0^{(j)} \leq \left[\wtL^d \rho_2\right] + 1 = : n_0^{max}$ 
depends only on $\rho_2$. 
This representation can be obtained as the result of a consecutive
division of the domain $\Lambda_N$ in sub-domains 
(each time we divide the domain in $2^d$ parts) 
until one obtains the domains $\Lambda_0^{(j)}$. 
In \eqref{eq:LDef}, choose $\wtL$ sufficiently large so that
$\left[\wtL^d \rho\right] \geq 1$, 
i.e., there is at least one particle in each sub-domain $\Lambda_0^{(j)}$. 

Let us denote by $S^{\star}$ the smallest number belonging to $\log{\bbN}$ 
that is larger than $S$:
\begin{equation}\label{eq:SstarDef}
S^{\star} = \inf\{Q \geq S, \exp{Q} \in \bbN\} \text{.}
\end{equation}
Then one calculates:
\begin{equation}\label{eqEBorne2}
 E_\omega(\Lambda_N, n_N, S_N) \leq \sum_{j = 1}^{2^{N d}} 
E_\omega\left(\Lambda_0^{(j)}, n_0^{(j)}, \left(S_N / 2^{N d}\right)^{\star}\right) 
+ \frac{A}{2} n_N n_0^{max} \sum_{m = 0}^{N - 1} 2^{(m + 1) d} R_m^{-2} \text{.}
\end{equation}
Since $S_N / 2^{N d} \leq \wtL^d \rho_2 \sigma_0$ and $n_0^{(j)} \leq n_0^{max}$, 
we can deduce that 
$E_\omega\left(\Lambda_0^{(j)}, n_0^{(j)}, \left(S_N / 2^{N d}\right)^{\star}\right)$ 
is bounded uniformly with respect to $N$ and $\omega$.
The proof is done by a trivial bounding of the potential
(as the number of terms in the potential is bounded) and 
by the application of Weyl asymptotic.
Dividing \eqref{eqEBorne2} by $2^{N d}$, we finish the proof.
\end{proof}

\begin{proof}[Proof of Proposition \ref{prop:DensProp} \eqref{it:densCont}]
Having established Proposition \ref{prop:DensLocalBound},
it is sufficient to apply a standard argument due to Jensen 
(see, for example, \cite{Polya_Szego_vol1}).
\end{proof}

\begin{corollary}
Suppose $\rho < \rho_c$ is of the form \eqref{eqRhoCond}, 
$\sigma_1$ and $\sigma_2$ are fixed.
Then the $L^2$-convergence
\begin{equation*}
n_N^{-1} E_\omega\left(\Lambda_N, 
2^{N d} \wtL^d \rho, 2^{N d} \wtL^d \rho \sigma\right) 
\to \densEn(\rho, \sigma)
\end{equation*}
is uniform in $\sigma \in [\sigma_1, \sigma_2]$
because the pointwise convergence of monotone functions to a
continuous function on a compact interval implies the uniform convergence 
(Dini's theorem).
\end{corollary}

From this corollary we deduce the following proposition, 
which weakens the restrictions on the way the entropy must go to infinity in 
the thermodynamic limit.
Instead of a very specifically chosen sequence $S_N$
as in Corollary \ref{corL2LimSuite}, 
we need only the linear dependence of entropy on the number of particles.

\begin{proposition}
Let $\rho > 0$ and $\sigma \geq 0$. 
Analogously to \eqref{eqnNSNdef}, construct a sequence of integers
$n_N = 2^{N d} \wtL^d \rho$. 
Let also $S_N$ be a sequence such that 
$\dfrac{S_N}{n_N} \xrightarrow[N \to +\infty]{} \sigma$. Then
\begin{equation*}
n_N^{-1} E_\omega(\Lambda_N, n_N, S_N) \xrightarrow[N \to +\infty]{L^2} \densEn(\rho, \sigma) \text{.}
\end{equation*}
\end{proposition}

\subsection{Proof of Theorem \ref{th:enConv} \eqref{it:CaseL2} ($L^2$-Convergence for General Domains)}

Now we are ready to show the existence of 
the thermodynamic $L^2$-limit in full generality.
First of all, we will establish that
\begin{equation}\label{eqGenDomUpperBound}
\limsup_{\Lambda \to \infty} \dfrac{E_\omega(\Lambda, n, S)}{n} 
\leq \densEn(\rho, \sigma) \text{.}
\end{equation}
Let $\rho_0 > \rho$ of the form \eqref{eqRhoCond} 
be close to $\rho$: 
\begin{equation}\label{eqRho0Repr}
 \rho_0 = \frac{n_0}{2^{N_0 d} \wtL^d}, \quad (n_0, N_0) \in \bbN^2 \text{.}
\end{equation}
The representation \eqref{eqRho0Repr} is not unique. 
Among all the representations for a fixed $\rho_0$
there exists the minimal one, i.e., where $N_0$ is minimal: 
$\rho_0 = n_{min} 2^{-N_{min} d} \wtL^{-d}$.
We write:
\begin{equation*}
 n = m n_0 + r_0, \quad 0 \leq r_0 < n_0\text{.}
\end{equation*}
For a fixed $\rho_0$ we choose $N_0$ in \eqref{eqRho0Repr} 
as a function of $n$ such that
\begin{equation}\label{eqN0choix}
 L_{N_0}^d / n \to 0, \quad L_{N_0}^\lambda / n \to +\infty \text{.}
\end{equation}
By Definition \ref{defFisherConv} 
a sufficiently large $\Lambda$ contains $m$ disjoint cubes
with sides $\xi L_{N_0}$, where
\begin{equation*}
 1 < \xi < \left(\frac{\rho_0}{\rho}\right)^{1 / d} \text{.}
\end{equation*}
Consequently, $\Lambda$ contains translated cubes 
$\Lambda_{N_0}^{(i)}$, $i = 1, \hdots, m$, 
at mutual distances at least $(\xi - 1) L_{N_0}$. 
By \eqref{eqEupperEst} one has
\begin{equation}\label{eqEnEst1}
 E_\omega(\Lambda, n, S) \leq \sum_{i = 1}^{m - 1} 
E_\omega(\Lambda_{N_0}^{(i)}, n_0, S / (m - 1)) 
+ E_\omega(\Lambda_{N_0}^{(m)}, n_0 + r_0, 0) 
+ \frac{A}{2} n^2 (\xi - 1)^{-\lambda} L_{N_0}^{-\lambda} \text{.}
\end{equation}
We treat now the term $E_\omega(\Lambda_{N_0}^{(m)}, n_0 + r_0, 0)$. 
By the domain division procedure similar to that described 
in Section \ref{subsectSuiteCubes}
we can reduce $\Lambda_{N_0}^{(m)}$ 
to the union of $2^{(N_0 - N_{min}) d} = n_0 / n_{min}$ 
(this is an integer) 
translates of $\Lambda_{N_{min}}$. We obtain
\begin{equation}\label{eqEN0est1}
 E_\omega(\Lambda_{N_0}^{(m)}, n_0 + r_0, 0) \leq 
\sum_{j = 1}^{n_0 / n_{min}} E_\omega(\Lambda_{N_{min}}^{(j)}, n_{min} + r_j, 0) 
+ \frac{A}{2} (2 n_0) \sum_{m = N_{min}}^{N_0 - 1} 2^{(m + 1) d} R_m^{-2}
\end{equation}
where $0 \leq r_j < n_{min}$ and $j \in \{1, \hdots, n_0 / n_{min}\}$. 
Clearly,
\begin{equation*}
 E_\omega(\Lambda_{N_{min}}^{(j)}, n_{min} + r_j, 0) \leq C_1 n_{min}^2
\end{equation*}
where the constant $C_1$ is uniform in $j$ and $\omega$. 
By \eqref{eqEN0est1} one obtains
\begin{equation*}
 E_\omega(\Lambda_{N_0}^{(m)}, n_0 + r_0, 0) \leq 
\frac{n_0}{n_{min}} C_1 n_{min}^2 + C_2 n_0 \leq C_3 n_{min} n_0 \text{.}
\end{equation*}
Finally, from \eqref{eqEnEst1} one deduces that
\begin{equation*}
n^{-1} E_\omega(\Lambda, n, S) \leq 
\frac{n_0}{n} \sum_{i = 1}^{m - 1} \frac{1}{n_0} 
E_\omega(\Lambda_{N_0}^{(i)}, n_0, S / (m - 1)) + C_3 \frac{n_{min} n_0}{n} 
+ \frac{A}{2} (\xi - 1)^{-\lambda} \frac{n}{L_{N_0}^{\lambda}} 
\text{.}
\end{equation*}
Note that $n_0 / n \leq 1 / m \to 0$ and ${n} / {L_{N_0}^{\lambda}} \to 0$ 
according to \eqref{eqN0choix}. Thus,
\begin{equation*}
\limsup_{\Lambda \to \infty} n^{-1} E_\omega(\Lambda, n, S) 
\leq \densEn\left(\rho_0, \sigma\right) \text{.}
\end{equation*}
Approaching $\rho$ from above by $\rho_0$, one gets \eqref{eqGenDomUpperBound}.

The proof that
\begin{equation*}
 \liminf_{\Lambda \to \infty} n^{-1} E_\omega(\Lambda, n, S) 
\geq \densEn(\rho, \sigma)
\end{equation*}
is done exactly as in \cite{Ruelle_StatMech}, pp.~47-48.
\qed

\subsection{$L^1$ and Almost Sure Limits}
\label{subsect:L1conv}

In this section we assume {\bf(LB)} and {\bf(Comp)}.
We show that if the interactions are compactly supported, 
the convergence to the thermodynamic limit can be improved.
The proof follows that of the $L^2$-convergence, 
so we only indicate the necessary modifications.

Let us introduce the following random variable. 
For all $n \in \bbN$, $S \in \bbR$ and domain $A \subset \bbR^d$ we set
\begin{equation*}
f_\omega(A, n, S) = E_\omega(\whA, n, S^\star)\text{,}
\end{equation*}
where $\whA$ is the $R_0 / 2$-interior of $A$, i.e., 
\begin{equation*} 
\whA = \{x \in A, \dist(x, \partial{A}) > R_0 / 2\} \text{,}
\end{equation*} 
and recall that $S^\star$ is defined in \eqref{eq:SstarDef}.
We make the following observations before giving a subadditivity condition.

\begin{lemma}\label{lmAuxil1}
Let $A$ and $B$ be two domains in $\calV$.
\begin{enumerate}
\item
If $A \subset B$, then $\whA \subset \whB$.
\item
If $A \cap B = \varnothing$, then $\dist(\whA, \whB) \geq R_0$.
\end{enumerate}
Let also $x, y \in \bbR$.
\begin{enumerate}
\item
If $x \leq y$, then $x^\star \leq y^\star$.
\item
$(x + y)^\star \leq x^\star + y^\star$.
\end{enumerate}
\end{lemma}

We modify now the subadditive inequality \eqref{eqEupperEst1}.

\begin{proposition}
Let the interactions $W$ be compactly supported {\bf(Comp)}. 
Let $n_1, n_2 \in \bbN$. 
If $A$ and $B$ are two disjoint domains in $\calV$, 
$A \cap B = \varnothing$, then 
\begin{equation}\label{eqSubAddCompSupp}
f_\omega(A \cup B, n_1 + n_2, S_1 + S_2) 
\leq f_\omega(A, n_1, S_1) + f_\omega(B, n_2, S_2) \text{.}
\end{equation}
\end{proposition}

\begin{proof}
As the interactions are compactly supported,
we get \eqref{eqEupperEst1} by the same manner as 
in Proposition \ref{propESubAdd} 
but without the interaction term (with $A = 0$). 
It remains to use Lemma \ref{lmAuxil1} 
and the monotonicity of energy with respect to entropy
in order to get \eqref{eqSubAddCompSupp}.
\end{proof}

Thanks to the subadditivity \eqref{eqSubAddCompSupp}, we prove the convergence
in $L^1$ and the almost sure convergence.

\begin{proof}[Proof of Theorem \ref{th:enConv} \eqref{it:CaseL1as}]
In order to prove this type of convergence, 
which is stronger than that in the part \eqref{it:CaseL2} of the theorem,
it is sufficient to modify Section \ref{subsectSuiteCubes}.
Everything what follows remains true without any modifications.

In Section \ref{subsectSuiteCubes} we change the definition of cubes 
$\Lambda_N$ by taking
\begin{equation*}
L_N = 2 \left[\frac{1}{2} \left(2^N \wtL - R_0 - \delta\right)\right]
\end{equation*}
in a place of \eqref{eqLNcubeCote}, where $\delta$ is a fixed positive 
constant.
This guarantees that one may put exactly $2^d$ translates of $\Lambda_N$ 
in a cube $\Lambda_{N + 1}$ at distances at least $R_0$ for a properly chosen
$\delta$.
The lower boundedness is given by 
\textbf{(LB)} and \textbf{(SI)}:
\begin{equation*}
\frac{f_\omega(A, n, S)}{n} \geq -B - C \text{.}
\end{equation*}
Next, we apply the multidimensional subadditive ergodic theorem 
(see, for example, \cite{Smythe_ThErgMultiParam}) and obtain the
$L^1$- and almost sure convergence of the sequence.
\end{proof}


\section{Free Particles}
\label{sect:FreeParticles}

As a complement, we study the thermodynamic limit 
for the energy density $\densEn(\rho, \sigma)$ 
in the case of free (noninteracting) particles:
\begin{equation}\label{eqNoInteractions}
 W \equiv 0 \text{.}
\end{equation}
We remark that the background potential $V_\omega$ remains present.
Interestingly, even in this case the results are not as trivial 
as one could have expected.
The obtained thermodynamic limits depend on quantum statistics.

\subsection{Maxwell-Boltzmann Particles}

For particles without statistics we establish the following theorem.

\begin{theorem}\label{thEnergW0}
Suppose that the interactions are absent \eqref{eqNoInteractions} 
and that the particles are of Maxwell - Boltzmann statistics.
Let $\Sigma$ be the almost sure spectrum of the one-particle Hamiltonian
$H_\omega(1)$. 
If
\begin{equation}\label{eqSpectrEss}
\Sigma = \supp{\rmd{N}} \text{,} 
\end{equation}
then
\begin{equation*}
 \densEn(\rho, \sigma) = \inf \Sigma
\end{equation*}
for all $\rho > 0$ and $\sigma \geq 0$.
\end{theorem}

\begin{remark}
The condition \eqref{eqSpectrEss} 
is satisfied under
rather general assumptions on the random potential $V_\omega$
(see, for example, \cite{Vesilic_IDSregularity}).
\end{remark}

In order to prove Theorem \ref{thEnergW0}
we will make use of two following lemmas. 
We assume that the conditions of this theorem are verified in the sequel.

\begin{lemma}\label{lemmeVPfix2Bottom}
Let $\omega$ be such that $\Spec(H_\omega(1)) = \Sigma$. 
Let $N \in \bbN$ be fixed. Then
\begin{equation*}
 E_N(H_\omega(\Lambda, 1)) \to \inf \Sigma, \quad \Lambda \to \infty \text{.}
\end{equation*}
\end{lemma}

\begin{proof}
Consider, as usual, Dirichlet boundary conditions.
Then for almost any $E \in \bbR$, 
\begin{equation*}
N_\Lambda(E) \nearrow N(E), \quad \Lambda \to \infty \text{,}
\end{equation*}
where $N_\Lambda$ is the pre-limit density of states, i.e., the counting 
function of the operator $H_\omega(\Lambda, 1)$ divided by $|\Lambda|$,
and $N$ is the density of states of the one-particle operator $H_\omega(1)$.

Just by the definition of the counting function
\begin{equation*}
N_\Lambda(E_N(H_\omega(\Lambda, 1))) = \frac{N}{|\Lambda|} \to 0, 
\quad \Lambda \to \infty \text{.}
\end{equation*}
Moreover, by the monotonicity of the Dirichlet eigenvalues,
$E_N(H_\omega(\Lambda, 1))$ decreases as $\Lambda \to \infty$
and thus necessarily converges:
\begin{equation*}
E_N(H_\omega(\Lambda, 1)) \searrow \wtE, \quad \Lambda \to \infty \text{.}
\end{equation*}
By combining the previous arguments, we find that
\begin{equation*}
N(\wtE) \nwarrow N_\Lambda(\wtE) \leq N_\Lambda(E_N(H_\omega(\Lambda, 1))) 
\to 0, \quad \Lambda \to \infty \text{,}
\end{equation*}
which implies
\begin{equation*}
N(\wtE) = 0 \text{.}
\end{equation*}
Consequently, 
\begin{equation*}
\wtE \leq \inf\Sigma \text{.}
\end{equation*}
Finally, $\wtE < \inf\Sigma$ is impossible, because necessarily
$E_N(H_\omega(\Lambda, 1)) \geq \inf\Sigma$.
\end{proof}

The next lemma is a modification of Proposition \ref{propESubAdd} 
and expresses a subadditive-type condition that is even stronger than
\eqref{eqEupperEst}.
It is exactly this lemma that is not valid for bosons and fermions.

\begin{lemma}
If $W \equiv 0$ and the particles under consideration are 
not restricted to any statistics, then
\begin{equation}\label{eqW0Subadd}
 E_\omega(\Lambda, n_1 + n_2, S_1 + S_2) 
\leq E_\omega(\Lambda, n_1, S_1) + E_\omega(\Lambda, n_2, S_2)
\text{.}
\end{equation}
\end{lemma}

\begin{proof}
The proof follows that of Proposition \ref{propESubAdd}.
The only modification is the construction of test functions in 
Lemma \ref{lemmeFonctTest}.
When the interactions are absent, one can place two groups of particles
in the same box and, consequently, one is not obliged to enlarge the size of
a box together with the number of particles.
\end{proof}

The last idea is not applicable to bosons or fermions, 
as one cannot guarantee the independence of constructed test functions
and the orthogonality of terms in
\eqref{eqZetaBosonicConstr} or \eqref{eqZetaFermionicConstr} 
is not assured.
Moreover, a constructed test function may happen to be identically zero 
for fermions.

\begin{proof}[Proof of Theorem \ref{thEnergW0}]
Because of the subadditivity \eqref{eqW0Subadd}, 
the following limit exists:
\begin{equation*}
 \exists \lim_{\substack{n \to \infty\\S / n \to \sigma\\\Lambda \text{ fixed}}}
\frac{E_\omega(\Lambda, n, S)}{n} =: \zeta(\Lambda, \sigma)
\end{equation*}
in the sense of $L^1$ and almost surely with respect to $\omega$.

This is proved exactly in the same manner as
the existence of $\densEn(\rho, \sigma)$ 
(and in some aspect is even simpler). 
As Dirichlet eigenvalues are monotonous with respect to the domain, 
the function $\zeta$ is nonincreasing in $\Lambda$ 
and nondecreasing in $\sigma$ (by obvious reasons). 
Due to the monotonicity in $\Lambda$, we find also that for $\rho > 0$ :
\begin{equation}\label{eqEboundZeta}
 \densEn(\rho, \sigma) \leq \lim_{\Lambda \to \infty}\zeta(\Lambda, \sigma) 
\text{.}
\end{equation}

We remark as well that by \eqref{eqW0Subadd}
the function $E_\omega(\Lambda, n, \sigma n) / n$ 
is nonincreasing in $n$. 
So one can interchange the two limits in the r.h.s. 
of \eqref{eqEboundZeta} to get
\begin{equation*}
 \densEn(\rho, \sigma) \leq 
\lim_{\Lambda \to \infty} \lim_{\substack{n \to \infty\\S / n \to \sigma}} 
\frac{E_\omega(\Lambda, n, S)}{n} = 
\lim_{\substack{n \to \infty\\S / n \to \sigma}} \lim_{\Lambda \to \infty} 
\frac{E_\omega(\Lambda, n, S)}{n} 
\leq \lim_{\substack{n \to \infty\\S / n \to \sigma}} 
\lim_{\Lambda \to \infty} E_\omega(\Lambda, 1, S / n) \text{.}
\end{equation*}
Here \eqref{eqW0Subadd} was used once more in the last inequality.
Finally, by Lemma \ref{lemmeVPfix2Bottom} we establish
\begin{equation*}
\lim_{\Lambda \to \infty} E_\omega(\Lambda, 1, S / n) \to \inf \Sigma
\end{equation*}
with $S$ and $n$ being fixed.
\end{proof}

Theorem \ref{thEnergW0} expresses the fact that 
the thermodynamic limit for Maxwell - Boltzmann particles 
is trivial in the absence of interactions.
Thus, it is indeed the interactions that may possibly render 
the limit being nontrivial.

\subsection{Bosons}

We remark that the energy levels for a system of noninteracting 
bosons coincide with energies for Maxwell - Boltzmann particles 
(see, for example \cite{LandauLifschitzQuantum, greiner2007quantum}).
On the contrary, 
for bosons the combinatorial degeneracy is lifted up by means of 
the symmetrization procedure
(the degeneracy due to coincidences like $E_2 + E_3 = E_1 + E_4$ remains).

Nevertheless, the ground state energy for bosons is the same as 
in the previous section, and, consequently, 
the ground state energy per particle converges in the thermodynamic limit 
to the lower edge of the almost sure spectrum of the one-particle operator.


\subsection{Fermions}

The situation changes significantly if particles are fermions.
For the basic properties of a system of noninteracting fermions, 
we refer the reader once more to 
\cite{LandauLifschitzQuantum, greiner2007quantum}.

For fermions we know only how to obtain results on the ground state energy 
$\densEn(\rho, 0)$.
The arguments we use do not rely on subadditivity properties and,
consequently, are valid for any boundary conditions.

The main difference between fermions and bosons is that 
the ground energy for $n$ noninteracting fermions is given by
the sum of the first $n$ energies of a one-particle system
\begin{equation}\label{eq:fermFondEn}
E_1(\Lambda, n) = \sum_{k = 1}^n E_k(\Lambda, 1)
\end{equation}
and not by $n$ times the one-particle ground energy.
The ground state itself is given by the Slater determinant
\begin{equation*}
\Omega_1(\Lambda, n) = \det{(\psi_i(x^j))_{i, j}} \text{,}
\end{equation*}
where $\psi_i$ is the eigenfunction of 
$H_\omega(\Lambda, 1)$ corresponding to the energy $E_i(\Lambda, n, \omega)$.

A comparison with the Laplacian and the use of Weyl asymptotic provide 
a simple proof that the limit $\densEn(\rho, 0)$ is strictly different from 
zero for $\rho > 0$ if the background potential is nonnegative.

\begin{proposition}
Suppose $V_\omega \geq 0$. 
Then there exists $\beta = \beta(d)$ such that
\begin{equation}\label{eq:FermFondEnBound}
\densEn(\rho, 0) \geq \beta \rho^{2 / d} \text{.}
\end{equation}
\end{proposition}

\begin{proof}
By the variational principle, as the potential is positive, we obtain
\begin{equation*}
E_j(\Lambda, 1) \geq E_j^0(\Lambda, 1) \text{,}
\end{equation*}
where $E_j^0$ are eigenvalues of $-\laplace$. Next, with the help of 
Weyl asymptotic for the Laplacian, by \eqref{eq:fermFondEn} we get
\begin{equation*}
\frac{E_1(\Lambda, n)}{n} = \frac{1}{n} \sum_{k = 1}^n E_k(\Lambda, 1) \geq 
\frac{1}{n} \sum_{k = 1}^n E_k^0(\Lambda, 1) \geq 
\frac{C_1}{n} \sum_{k = 1}^n \left(\frac{k}{|\Lambda|}\right)^{2 / d} \geq 
\frac{C_2}{n |\Lambda|^{2 / d}} \cdot n^{2 / d + 1} \geq C_3 \rho^{2 / d} 
\text{,}
\end{equation*}
which proves \eqref{eq:FermFondEnBound}.
\end{proof}

\begin{remark}
The generalization to the case of lower-bounded random potential is obvious.
\end{remark}

We now compute an explicit expression for the limit
\begin{equation*}
\densEn(\rho, 0) = \lim_{\substack{\Lambda \to \infty\\n / |\Lambda| \to \rho}} 
\frac{E_1(\Lambda, 1) + \hdots + E_n(\Lambda, n)}{n}
\end{equation*}
in terms of the integrated density of states of the one-particle problem.
Once more for simplicity, we suppose that background potential is nonnegative: 
$V_\omega \geq 0$.
We need only rather general assumptions on the density of states, 
which we denote by $N(E)$.
\begin{condition}\label{cond:DensEtats}
The integrated density of states $N(E)$ is a continuous function
and defines a positive measure ${\rmd}N(E)$ 
such that the almost sure spectrum is equal to the support of ${\rmd}N$:
\begin{equation*}
\supp{{\rmd}N} = \Sigma \text{.}
\end{equation*}
\end{condition}

\begin{remark}
The last condition is certainly verified, for example, 
if the Wegner estimate \textbf{(W)}
holds for $H_\omega(\Lambda, 1)$.
\end{remark}

\begin{proposition}[Wegner estimate]\label{prop:WegnerEstimate}
Let $\Sigma$ be almost sure spectrum of $H_\omega$.
There exists constant $C > 0$ such that, for any Borel subset $I \subset \bbR$,
$$
\bbE\left(\Tr\left(\ds1_I(H_\omega(\Lambda))\right)\right) 
\leq C |\Lambda| \cdot |I \cap \Sigma| \text{.}
\eqno{\bf(W)}
$$
\end{proposition}

\begin{proof}
Wegner estimate is well known for both discrete and continuous Anderson model 
under the assumption that the random variables are i.i.d. and that their 
distribution is regular 
\cite{CombesHislopKlopp_WegnerEstimate, Vesilic_IDSregularity}.
\end{proof}

\begin{definition}\label{def:FermiEn}
Fix a density of particles $\rho$. 
The \emph{Fermi energy} $E_\rho$ is a solution of the equation
\begin{equation}\label{eq:FermiEn}
N(E_\rho) = \rho \text{.}
\end{equation}
\end{definition}

\begin{remark}
It may happen that $\sup_E{N(E)} < \rho$. 
For example, if one considers a discrete Anderson model, 
then $N(E) \leq 1$ and for $\rho > 1$ the equation \eqref{eq:FermiEn} 
does not have any solutions.
This is due to the fact that the density of particles is too big 
(in other words, there isn't enough space for so many particles) 
to accommodate for $n$ fermions.
This situation never arises in a continuous setting.
\end{remark}

\begin{remark}
A solution of the equation \eqref{eq:FermiEn} is not necessarily unique if
the integrated density of states is flat on the level $\rho$. 
As $N(E)$ is a continuous nondecreasing function, the set of solutions is 
the closed interval $[E_\rho^{min}, E_\rho^{max}]$. 
From the spectral point of view the open interval 
$(E_\rho^{min}, E_\rho^{max})$ doesn't play any role 
because its intersection with the almost sure spectrum $\Sigma$ is empty.

In this situation we will also use the notation introduced by the 
Definition \ref{def:FermiEn}, meaning 
$E_\rho =  [E_\rho^{min}, E_\rho^{max}]$. 
As we will see, this convention is consistent with the results.
\end{remark}

Next theorem is the main result of this section.

\begin{theorem}\label{th:densEnFerm}
Let $\rho > 0$. Then
\begin{equation}\label{eq:densEnFerm}
\densEn(\rho, 0) = \frac{1}{\rho} \int_0^{E_\rho} E {\rmd}N(E) \text{.}
\end{equation}
\end{theorem}

To give a proof to this theorem, we will need the following crucial lemma, 
which explains why $E_\rho$ given by \eqref{eq:FermiEn} corresponds exactly 
to the common physical notion of the 
Fermi energy.\footnote{Fermi energy is the energy of highest occupied 
quantum state in a system of fermions at absolute zero temperature.
Alternatively, for non-interacting fermions, 
it is the increase in the ground state energy when one particle is added to 
the system.
For more details on the concept on the Fermi energy, 
see \cite{AshcroftMermin_SSP}.}

\begin{lemma}\label{lem:convEnFerm}
\begin{equation}\label{eq:convEnFerm}
E_n(\Lambda, 1, \omega) 
\xrightarrow[\substack{\Lambda \to \infty\\n / |\Lambda| \to \rho}]
{\omega \text{-p.s.}} E_\rho \text{.}
\end{equation}
\end{lemma}

\begin{proof}
We denote by $N_\omega^\Lambda$ the density of states 
before taking the limit for a one-particle operator:
\begin{equation*}
N_\omega^\Lambda(E) 
= \frac{\fCount(E, H_\omega(\Lambda, 1))}{|\Lambda|} \text{.}
\end{equation*}
Then by definition, in the thermodynamic limit
\begin{equation*}
N_\omega^\Lambda(E_n(\Lambda, 1, \omega)) = \frac{n}{|\Lambda|} \to \rho \text{.}
\end{equation*}
On the other hand, by the existence of the integrated density of states 
we get: 
\begin{equation*}
N_\omega^\Lambda(\xi) 
\xrightarrow[\Lambda \to \infty]{\omega \text{-p.s.}}
N(\xi)
\quad \forall \xi \in \bbR
\text{.}
\end{equation*}
We finish the proof by applying the monotonicity argument. 
Suppose that
\begin{equation*}
\liminf{E_n(\Lambda, 1, \omega)} < E_\rho^{min} \text{.}
\end{equation*}
By passing to a subsequence, we find that
there exists $\delta > 0$ such that  
$E_n(\Lambda, 1, \omega) < E_\rho^{min} - \delta$. 
We arrive to a contradiction: 
\begin{equation*}
N(E_\rho^{min}) 
\gets N_\omega^\Lambda(E_n(\Lambda, 1, \omega)) 
\leq  N_\omega(\Lambda(E_\rho^{min} - \delta)) 
\to N(E_\rho^{min} - \delta) 
< N(E_\rho^{min}) \text{.}
\end{equation*}
The last inequality is strict because $E_\rho^{min}$ 
is the minimal value of energy such that $N(E) = \rho$ 
and so for any $E$ above this level 
the density of states $N(E)$ is strictly smaller.

Similarly, we show that 
$\limsup{E_n(\Lambda, 1, \omega)} \leq E_\rho^{max}$, so
\begin{equation*}
E_\rho^{min} \leq \liminf{E_n(\Lambda, 1, \omega)} \leq 
\limsup{E_n(\Lambda, 1, \omega)} \leq E_\rho^{max} \text{,}
\end{equation*}
which is equivalent to \eqref{eq:convEnFerm}.
\end{proof}

\begin{proof}[Proof of Theorem \ref{th:densEnFerm}]
To show \eqref{eq:densEnFerm} we write
\begin{multline*}
\frac{E_1(\Lambda, n, \omega)}{n} 
= \frac{E_1(\Lambda, 1, \omega) + \hdots + E_n(\Lambda, 1, \omega)}{n} 
= \frac{1}{n} \Tr\left[H_\omega(\Lambda, 1) 
\cdot \ds1_{[0, E_n(\Lambda, 1, \omega)]}(H_\omega(\Lambda, 1))\right] \\
= \frac{|\Lambda|}{n} \int_0^{E_n(\Lambda, 1, \omega)} E {\rmd}N_\omega^\Lambda(E) 
\xrightarrow[\substack{\Lambda \to \infty\\n / |\Lambda| 
\to \rho}]{\omega \text{-p.s.}} 
\frac{1}{\rho} \int_0^{E_\rho} E {\rmd}N(E) \text{,}
\end{multline*}
where the convergence is valid because the measure ${\rmd}N_\omega^\Lambda$ 
converges weakly to ${\rmd}N$, 
the integration limit converges to $E_\rho$ by Lemma \ref{lem:convEnFerm}
and the dominated convergence theorem can be applied.
\end{proof}

\begin{remark}
The formula \eqref{eq:densEnFerm} admits an alternative form: 
\begin{equation*}
\densEn(\rho, 0) = 
\frac{\int_0^{E_\rho} E {\rmd}N(E)}{\int_0^{E_\rho} {\rmd}N(E)} \text{,}
\end{equation*}
which reads as the ground state energy density is the energy averaged
from zero to the Fermi energy with respect to the density of states.
\end{remark}

\bibliographystyle{alpha}
\bibliography{rechercheBiblio}

\end{document}